\pdfoutput=1
\documentclass[acmsmall,screen]{acmart}

\setcopyright{rightsretained}
\acmPrice{}
\acmDOI{10.1145/3591291}
\acmYear{2023}
\copyrightyear{2023}
\acmSubmissionID{pldi23main-p578-p}
\acmJournal{PACMPL}
\acmVolume{7}
\acmNumber{PLDI}
\acmArticle{177}
\acmMonth{6}
\received{2022-11-10}
\received[accepted]{2023-03-31}

\startPage{1}

\usepackage{booktabs}   
\usepackage{subcaption} 

\usepackage{multirow, makecell}
\usepackage[utf8]{inputenc} 
\usepackage{amsmath, amsthm}
\usepackage{thmtools}
\usepackage[english]{babel}
\usepackage{pdfpages}
\usepackage{paralist}
\usepackage{graphicx}
\usepackage{subcaption}
\usepackage{caption}
\usepackage{longtable}
\usepackage{booktabs}
\usepackage{tabu}
\usepackage[referable]{threeparttablex}
\usepackage{varwidth}
\usepackage{multirow}
\usepackage{wrapfig}
\usepackage{enumerate}
\usepackage[inline]{enumitem}
\usepackage{microtype}
\usepackage{xifthen}
\usepackage[normalem]{ulem}
\usepackage{versions}
\usepackage{siunitx}
\usepackage{dsfont}

\usepackage{pifont}

\usepackage[ruled,vlined,resetcount,linesnumbered]{algorithm2e}

\usepackage{parskip}

\usepackage{hyperref}
\usepackage{thm-restate}
\usepackage[capitalise]{cleveref}

\crefname{figure}{Figure}{Figure}
\crefformat{footnote}{#2\footnotemark[#1]#3}
\usepackage{comment}

\usepackage{todonotes}
\usepackage{tikz}
\usetikzlibrary{arrows,automata,shapes,decorations,decorations.markings,calc, matrix,decorations.pathmorphing, patterns,backgrounds,shapes.misc,arrows.meta}

\usepackage{bm}
\usepackage{pgfplots}
\usepackage{enumerate,paralist}
\usepackage{pbox}

\usepackage{appendix}
\usepackage{calc}
\usepackage{fp}
\usepackage{multirow}
\usepackage{pbox}

\usepackage[mathscr]{euscript}
\usepackage{mathtools}

\usepackage{footnote}
\usepackage{bbm}
\usepackage{multicol}
\usepackage{graphicx} 
\usepackage{subcaption}

\usepackage{listings, xcolor}
\include{macros}

\renewcommand{\smallskip}{}
\excludeversion{pldi}
\includeversion{arxiv}

\begin{document}

\title{Sound Dynamic Deadlock Prediction in Linear Time}


\author{H\"{u}nkar Can Tun\c{c}}
\orcid{0000-0001-9125-8506}
\affiliation{
	\institution{Aarhus University}            
	\country{Denmark}                    
}
\email{tunc@cs.au.dk}          

\author{Umang Mathur}
\orcid{0000-0002-7610-0660}
\affiliation{
  \institution{National University of Singapore}            
  \country{Singapore}                    
}
\email{umathur@comp.nus.edu.sg}          

\author{Andreas Pavlogiannis}
\orcid{0000-0002-8943-0722}
\affiliation{
	\institution{Aarhus University}            
	\country{Denmark}                    
}
\email{pavlogiannis@cs.au.dk}          

\author{Mahesh Viswanathan}
\orcid{0000-0001-7977-0080}
\affiliation{
	\institution{University of Illinois, Urbana Champaign}            
	\country{USA}                    
}
\email{vmahesh@illinois.edu}

\begin{abstract}
Deadlocks are one of the most notorious concurrency bugs, 
and significant research has focused on detecting them efficiently.
\emph{Dynamic predictive analyses} work by observing concurrent executions,
and reason about alternative interleavings that can witness concurrency bugs.
Such techniques offer scalability and sound bug reports, and have emerged as an
effective approach for concurrency bug detection, such as data races.
Effective dynamic deadlock prediction, however, has proven a challenging task,
as no deadlock predictor currently meets the requirements of 
\emph{soundness}, \emph{high-precision}, and \emph{efficiency}.

In this paper, we first formally establish that this tradeoff is unavoidable,
by showing that 
(a) sound and complete deadlock prediction is intractable, in general,
and (b) even the seemingly simpler task of determining the presence of \emph{potential} deadlocks,
which often serve as unsound witnesses for actual predictable deadlocks, is intractable.
The main contribution of this work is a new class of predictable deadlocks, 
called \emph{sync(hronization)-preserving} deadlocks.
Informally, these are deadlocks
that can be predicted by reordering the observed execution while preserving the relative order of conflicting critical sections.
We present two algorithms for \emph{sound} deadlock prediction based on this notion.
Our first algorithm~\SyncPDOffline detects all sync-preserving deadlocks, with running time that is linear per \emph{abstract deadlock pattern}, a novel notion also introduced in this work.
Our second algorithm~\SyncPDOnline predicts all sync-preserving deadlocks
that involve two threads in a \emph{strictly online} fashion, runs in overall linear time, 
and is better suited for a runtime monitoring setting.
%
%

We implemented both our algorithms and evaluated their ability to perform offline and online deadlock-prediction on a large dataset of standard benchmarks.
Our results indicate that our new notion of sync-preserving deadlocks is highly effective, as
(i)~it can characterize the vast majority of deadlocks and
(ii)~it can be detected using an online, sound, complete and highly efficient algorithm.
\end{abstract}




\maketitle


\section{Introduction}
\seclabel{intro}


The verification of concurrent programs is a major challenge due to the non-deterministic behavior 
intrinsic to them.
Certain scheduling patterns may be unanticipated by the programmers,
which may then lead to introducing concurrency bugs.
Such bugs are easy to introduce during development but can be very hard to
reproduce during in-house testing, and have been notoriously called \emph{heisenbugs}~\cite{Musuvathi2008}.
Among the most notorious concurrency bugs are deadlocks, occurring when the system blocks its execution because each thread is waiting for another thread to finish a task in a circular fashion.
Deadlocks account for a large fraction of concurrency bugs in the wild across various programming languages~\cite{Lu08,Tu2019}
while they are often introduced accidentally when fixing other concurrency bugs~\cite{Yin2011}.

Deadlock-detection techniques can be broadly classified into static and dynamic techniques.
As usual, static techniques analyze source code and have the potential to prove the absence of 
deadlocks~\cite{Naik2009,Ng2016,Liu2021}.
However, as static analyses face simultaneously two dimensions of non-determinism, namely in inputs and scheduling, they  lead to poor performance in terms of scalability and false positives,
and are less suitable when the task at hand is to help software developers proactively find bugs.
Dynamic analyses, on the other hand, have the more modest goal of discovering deadlocks by analyzing program executions, allowing for better scalability and few (or no) false positives.
Although dynamic analyses cannot prove the absence of bugs, 
they offer \emph{statistical} and \emph{coverage} guarantees.
These advantages have rendered dynamic techniques a standard practice in 
principled testing for various bugs,
such as data races, atomicity violations, deadlocks, and 
others~\cite{Flanagan09,threadsanitizer,Bensalem2005,Flanagan2008,Mathur2020,Biswas14,Savage97,Pozniansky03}.
A recent trend in this direction advocates for 
\emph{predictive analysis}~\cite{Smaragdakis12,Huang14,Kini2017,Flanagan2008,Huang2018,Kalhauge2018,Genc19},
where the goal is to enhance coverage by additionally reasoning about alternative
reorderings of the observed execution trace that \emph{could} have taken place and also manifest the bug.

Due to the difficulty of the problem, many dynamic deadlock analyses focus on detecting \emph{deadlock patterns}, broadly defined as cyclic lock-acquisition patterns in the observed execution trace.
One of the earliest works in this direction is the Goodlock algorithm~\cite{Havelund2000}.
As deadlock patterns are necessary but insufficient conditions for the presence of deadlocks, subsequent work has focused on refining this notion in order to reduce false-positives~\cite{Bensalem2005,Agarwal2005}.
Further techniques reduce the size of the lock graph to improve scalability~\cite{Cai2012,Cai2020}.
To further address the unsoundness (false positives) problem, various works propose controlled-scheduling techniques that attempt to realize deadlock warnings via program re-execution~\cite{Bensalem2006,Joshi2009,Samak2014,Samak2014b,Sorrentino2015}
and exhaustive exploration of all reorderings~\cite{Joshi2010,Koushik05}.

Fully sound deadlock prediction has traditionally relied
on explicitly~\cite{Joshi2010,Koushik05} or symbolically (SMT-based)~\cite{Eslamimehr2014,Kalhauge2018} producing all sound witness reorderings.
The heavyweight nature of such techniques
limits their applicability to executions of realistic size, 
which is often in the order of millions of events.
The first steps for sound, polynomial-time deadlock prediction were made recently with 
\seqc~\cite{Cai2021}, an extension of M2~\cite{Pavlogiannis2020} that targets data races.

This line of work highlights the need for a most-efficient sound deadlock predictor, approaching the golden standard of \emph{linear time}.
Moreover, dynamic analyses are often employed as runtime monitors, and must thus operate \emph{online}, reporting bugs as soon as they occur. Unfortunately, most existing online algorithms only report \emph{deadlock patterns}, 
thus suffering false positives.
The lack of such a deadlock predictor is even more pronounced when contrasted to 
the problem of dynamic race prediction, which has 
seen a recent surge of sound, online, \emph{linear-time} predictors~(e.g.,~\cite{Kini2017,Roemer20}), and highlights the bigger challenges that deadlocks entail.
We address these challenges in this work, 
by presenting the first high-precision, sound dynamic deadlock-prediction algorithm 
that works online and in linear time.


The task of checking if a potential deadlock is a real predictable
deadlock, in general, involves searching for the reordering of the original execution 
that witnesses the deadlock.
The first ingredient towards our technique is the notion of 
\emph{synchronization-preserving reorderings}~\cite{Mathur2021} that help systematize this search space.
\emph{Synchronization-preserving deadlocks} are then those 
predictable deadlocks that can be witnessed in some synchronization-preserving reordering.
We illustrate synchronization-preserving deadlocks using an example in~\cref{subsec:spd_intro}.

This notion of synchronization-preservation, by itself, is not sufficient
when it comes to deadlock detection as the prerequisite step towards
predicting deadlocks also involves identifying  \emph{potential deadlock patterns}.
Unlike data races, where \emph{potential races}
can be identified in polynomial-time, the identification of deadlock patterns
is in general, intractable; we prove this in~\secref{lower-bounds}.
As a result, an approach that works by explicitly enumerating
cycles in a \emph{lock graph} and then checking if any of these cycles is realizable 
to a deadlock is likely to be not scalable.
To tackle this, we propose the novel notion of \emph{abstract deadlock patterns}
which, informally, represent clusters of deadlock patterns of the same signature. 
Intuitively, a set of deadlock patterns have the same signature 
if the threads and locks that participate in the patterns are the same.
Our next \emph{key observation} is that a single abstract deadlock pattern
can be checked for sync-preserving deadlocks in \emph{linear total time} in the length of the execution, 
\emph{regardless} of how many concrete deadlock patterns it represents.
Our first deadlock prediction algorithm \SyncPDOffline builds upon this --- 
it enumerates all abstract deadlock patterns
in a first phase and then checks their realizability in a second phase,
while running in linear time per abstract deadlock pattern.
Since the number of abstract deadlock patterns
is typically \emph{far smaller} than the number of (concrete)
deadlock patterns (see \cref{tab:expr-results} in \secref{experiments}), this approach achieves high scalability.
Our second algorithm \SyncPDOnline works in a single streaming pass --- it computes abstract deadlock patterns
that involve only two threads and checks their realizability \emph{on-the-fly} simultaneously in overall linear time in the length of the execution.


\subsection{Synchronization-Preserving Deadlocks}
\label{subsec:spd_intro}



\begin{figure}[t]
\begin{subfigure}[t]{0.22\textwidth}
\scalebox{0.855}{
\execution{2}{
\figev{1}{$\acq(\LockColorOne{\lk_1})$}
\figev{1}{$\Bacq(\LockColorTwo{\lk_2})$}
\figev{1}{\underline{$\wt(x)$}}
\figev{1}{$\rel(\LockColorTwo{\lk_2})$}
\figev{1}{$\rel(\LockColorOne{\lk_1})$}
\figev{2}{$\acq(\LockColorTwo{\lk_2})$}
\figev{2}{\underline{$\rd(x)$}}
\figev{2}{$\Bacq(\LockColorOne{\lk_1})$}
\figev{2}{$\rel(\LockColorOne{\lk_1})$}
\figev{2}{$\rel(\LockColorTwo{\lk_2})$}
}
}
\caption{
A trace $\tr_1$ with no\\ predictable deadlock.
} 
\label{fig:motivating-no-dl}
\end{subfigure}
~~\hspace{-0.1cm}
\begin{subfigure}[t]{0.75\textwidth}
\scalebox{0.855}{
\execution{4}{
\figev{1}{$\acq(\LockColorOne{\lk_1})$}
\figev{1}{$\rel(\LockColorOne{\lk_1})$}
\figev{2}{\underline{$\acq(\LockColorTwo{\lk_2})$}}
\figev{2}{$\mathbf{\Bacq(\LockColorThree{\lk_3})}$}
\figev{2}{$\wt(z)$}
\figev{2}{$\rel(\LockColorThree{\lk_3})$}
\figev{2}{$\rel(\LockColorTwo{\lk_2})$}
\figev{4}{$\acq(\LockColorOne{\lk_1})$}
\figev{4}{\underline{$\wt(y)$}}
\figev{4}{$\rd(z)$}
}
\execution{4}{
	\figevoffset{10}{4}{$\rel(\LockColorOne{\lk_1})$}
	\figevoffset{10}{1}{$\acq(\LockColorThree{\lk_3})$}
	\figevoffset{10}{1}{$\wt(x)$}
	\figevoffset{10}{1}{$\rd(y)$}
	\figevoffset{10}{1}{\underline{$\rel(\LockColorThree{\lk_3})$}}
	\figevoffset{10}{3}{$\acq(\LockColorThree{\lk_3})$}
	\figevoffset{10}{3}{\underline{$\rd(x)$}}
	\figevoffset{10}{3}{$\Bacq(\LockColorTwo{\lk_2})$}
	\figevoffset{10}{3}{$\rel(\LockColorTwo{\lk_2})$}
	\figevoffset{10}{3}{$\rel(\LockColorThree{\lk_3})$}
}
}
\caption{
A  trace $\tr_2$ with a sync-preserving deadlock, stalling $t_2$ on $e_4$ and $t_3$ on $e_{18}$.
}
\label{fig:motivating-dl}
\end{subfigure}
\caption{
Traces with no predictable deadlock (\subref{fig:motivating-no-dl}), and with a sync-preserving deadlock (\subref{fig:motivating-dl}).
} 
\label{fig:motivating}
\end{figure}


Consider the trace $\tr_1$ in \cref{fig:motivating-no-dl} consisting of 
$10$ events and two threads.
We use $e_i$ to denote the $i$-th event of $\tr_1$.
The events $e_2$ and $e_{8}$ form a \emph{deadlock pattern}:~
they respectively acquire the locks $\LockColorTwo{\lk_2}$ and $\LockColorOne{\lk_1}$ while holding the locks $\LockColorOne{\lk_1}$ and $\LockColorTwo{\lk_2}$,
and no common lock protects these operations.

A deadlock pattern is a necessary but insufficient condition for an actual deadlock:~a sound algorithm must examine whether it can be realized to a deadlock via a witness.
A witness is a reordering $\rho$ of (a slice of) $\tr_1$ 
that is also a valid trace, and such that $e_2$ and $e_{8}$ are locally enabled in their respective threads at the end of $\rho$.
In general, the problem of checking if a deadlock pattern
can be realized is intractable
(\thmref{w1-hardness-pattern}).
In this work we focus on checking whether a given deadlock pattern forms a
\emph{sync-preserving deadlock}, which is a subclass of the class of
all predictable deadlocks.

A deadlock pattern is said to be sync-preserving deadlock
if it can be witnessed in a \emph{sync-preserving reordering}.
A reordering $\rho^{\sf SP}$ of a trace $\tr$ is said to be sync-preserving if 
it preserves the control
flow taken by the original observed trace $\tr$, and further
it preserves the mutual order of any two critical sections (on the same lock)
that appear in the reordering $\rho^{\sf SP}$.
Consider, for example, the sequence $\rho_1 = e_1..e_3\,e_6..e_7$ where $e_i..e_j$ denote the contiguous sequence of events
that starts from $e_i$ and ends at $e_j$. 
We call $\rho_1$ a \emph{correct reordering} of $\tr_1$, being a slice of $\tr_1$ closed under the thread order and preserving the writer of each read in $\tr_1$;
the precise definition is presented in \secref{prelim}.
In this case, however, $\rho_1$ does not witness the deadlock as the event $e_2$ is not \emph{enabled} in $\rho_1$.
In fact, due to the dependency between the events $e_3$ and $e_7$, there are no correct reorderings of $\tr_1$ which make both $e_2$ and $e_8$ enabled.
This makes the deadlock pattern $\pattern{e_2, e_{8}}$ non-predictable.
Consider now $\tr_2$ in \cref{fig:motivating-dl}, and the sequence $\rho_2 = e_3..e_7\,e_8..e_{11}\,e_1 e_2$.
Observe that $\rho_2$ is also a correct reordering.
However, $\rho_2$ is not sync-preserving as the order of 
the two critical sections
on lock $\lk_1$ in $\rho_2$ is different from their original order in $\tr_2$.
On the other hand, $\rho_3 = e_1 e_2 e_3 e_8 e_9 \,e_{12}..e_{15}\,e_{16} e_{17}$
is a correct reordering that is also sync-preserving --- all pairs of critical 
sections on the same lock appear in the same order in $\rho_3$ as they did in $\tr_2$.
Further, $\rho_3$ also witnesses the deadlock as the events $e_4$ and $e_{18}$
are both \emph{enabled} in $\rho_3$.
This makes the deadlock pattern $\pattern{e_4, e_{18}}$ a sync-preserving deadlock.

In this work we show that sync-preserving deadlocks enjoy two remarkable properties.
First, all sync-preserving deadlocks 
of a given \emph{abstract} deadlock pattern can be checked in linear time.
Second, our extensive experimental evaluation on standard benchmarks indicates that 
sync-preservation captures 
a vast majority of deadlocks in practice.
In combination, these two benefits suggest that sync-preservation is the 
right notion of deadlocks to be targeted by dynamic deadlock predictors.


\subsection{Our Contributions}\seclabel{contributions}
In detail, the contributions of this work are as follows.
\begin{enumerate}[label=(\arabic*)]
\item {\bf Complexity of Deadlock Prediction.} Perhaps surprisingly, the complexity of detecting deadlock patterns, as well as predicting deadlocks has remained elusive.
Our first contribution resolves such questions. 
Given a trace $\tr$ of size $\NumEvents$ and $\NumThreads$ threads,
we first show that detecting even one deadlock \emph{pattern} of length $k$ is $\W{1}$-hard in $k$.
This establishes that the problem is $\NP$-hard, and 
further rules out algorithms that are fixed-parameter-tractable in $k$, i.e.,
with running time of the form $f(k)\cdot \poly{\NumEvents}$, for some function $f$.
We next show
that even with just $\NumThreads=2$ threads, the problem of 
detecting a single deadlock \emph{pattern} (of size $k=2$) admits a
quadratic lower bound, i.e., it cannot be solved in time $O(\NumEvents^{2-\epsilon})$,
no matter what $\epsilon > 0$ we choose.
These two results shed light on the difficulty in identifying deadlock patterns --- a task that might otherwise appear easier than the core task of prediction. 
These hardness results, in particular the fine-grained lower bound result,
are based on novel constructions, and results from fine-grained complexity~\cite{williams2018}.
Our third result is about confirming predictable deadlocks ---
even for a deadlock pattern of size $k=2$,
checking whether it yields a \emph{predictable deadlock} is $\W{1}$-hard
in the number of threads $\NumThreads$ (and thus again $\NP$-hard), 
and is inspired from an analogous result in the context of data race prediction~\cite{Mathur2020b}.
These results capture the intractability of deadlock prediction in general, 
even for the class of parametrized algorithms.

\item {\bf Sync-preserving Deadlock Prediction and Abstract Deadlock Patterns.}
Given the above hardness of predicting arbitrary deadlocks,
we define a novel notion of sync(hronization)-preserving deadlocks, illustrated in \cref{subsec:spd_intro}.
We develop $\SyncPDOnline$, an \emph{online, sound} deadlock predictor that takes as input a trace and reports \emph{all} sync-preserving deadlocks of size $2$ in \emph{linear time} $\Otilde(\NumEvents)$\footnote{We use $\Otilde$ to ignore \emph{polynomial} appearance of trace parameters typically much smaller than $\NumEvents$ (e.g., number of threads).}.
As most deadlocks in practice involve only two threads~\cite{Lu08}, restricting $\SyncPDOnline$ to size $2$ deadlocks leads to linear-time deadlock prediction with small impact on its coverage.
We also develop our more general algorithm, $\SyncPDOffline$,
that detects \emph{all} sync-preserving deadlocks of all sizes.
$\SyncPDOffline$ operates in two phases.
In the first phase, it detects all
\emph{abstract deadlock patterns}.
An abstract deadlock pattern is a novel notion that serves as a succinct representation
of the class of deadlock patterns having the same signature.
%
In the second phase, $\SyncPDOffline$ executes 
$\SyncPDOnline$ on each abstract pattern to decide whether a deadlock is formed.
The running time of $\SyncPDOffline$ remains linear in $\NumEvents$, 
but increases by a factor proportional to the number 
of abstract deadlock patterns in the lock graph.

\item {\bf Implementation and Evaluation.}
We have evaluated $\SyncPDOnline$ and $\SyncPDOffline$ in terms of performance and predictive power on a large dataset of standard benchmarks.
In the offline setting, $\SyncPDOffline$ finds the same number of deadlocks as the recently introduced \seqc,
 while achieving a speedup of $>200\times$ on the most demanding benchmarks, and $21\times$ overall.
 In the online setting, $\SyncPDOnline$ achieved a significant improvement in deadlock discovery and deadlock-hit-rate compared to the random scheduling based controlled concurrency testing technique of \dlfuzzer~\cite{Joshi2009}.
Our experiments thus support that the notion of sync-preserving deadlocks is suitable:
(i)~it captures the vast majority of the 
deadlocks in practice, and
(ii)~sync-preserving deadlocks can be detected online and optimally --- that is, soundly, completely and in linear time,
(iii)~it can enhance the deadlock detection capability of controlled concurrency testing techniques, 
(iv)~with reasonable runtime overhead.
\end{enumerate}


\section{Preliminaries}
\seclabel{prelim}

Here we set up our model and develop relevant notation,
following related work in predictive analyses of concurrent programs~\cite{Smaragdakis12,Kini2017,Roemer20}.

\myparagraph{Execution traces}{
	A dynamic analysis observes traces
	generated by a concurrent program,
	and analyzes them to determine the presence of a bug.
	Each such trace $\tr$ is a linear arrangement of events $\events{\tr}$.
	An event $e \in \events{\tr}$ is tuple $e = \ev{i, t, o}$, 
	where $i$ is a unique identifier of $e$,
	$t$ is the unique identifier of the thread performing $e$, and $o$
	is either a read or write ($o = \rd(x)$ or $o = \wt(x)$)
	operation to some variable $x$,
	or an acquire or release ($o = \acq(\lk)$ or $o = \rel(\lk)$)
	operation on some lock $\lk$.
	For the sake of simplicity, we often omit $i$ when referring to an event.
	We use $\ThreadOf{e}$ and $\OpOf{e}$ to respectively denote
	the thread identifier and the operation performed in the event $e$.
	We use $\threads{\tr}$, $\vars{\tr}$ and $\locks{\tr}$
	to denote the set of thread, variable and lock identifiers
	 in $\tr$.

	We restrict our attention to \emph{well-formed} traces $\tr$, that abide to shared-memory semantics.
	That is, if a lock $\lk$
	is acquired at an event $e$ by thread $t$, then any later acquisition
	event $e'$ of the same lock $\lk$
	must be preceded by an 
	event $e''$ that releases lock $\lk$ in thread $t$
	in between the occurrence of $e$ and $e'$.
	Taking $e''$ to be the earliest such release event,
	we say that $e$ and $e''$ are matching acquire and 
	release events, and denote this by $e = \match{\tr}(e'')$ and $e'' = \match{\tr}(e)$.
	Moreover, every read event has at least one preceding write event on the same location, that it reads its value from.
}

\myparagraph{Functions and relations on traces}{
	A trace $\tr$ implicitly defines some relations.
	The \emph{trace-order} $\trord{\tr} \subseteq \events{\tr}\times\events{\tr}$ 
	orders the events of $\tr$ in a total order based
	on their order of occurrence in the sequence $\tr$.
	The \emph{thread-order} $\tho{\tr}$
	is the unique partial order over $\events{\tr}$ 
	such that $e \tho{\tr} e'$
	iff $\ThreadOf{e} = \ThreadOf{e'}$ and $e \trord{\tr} e'$.
	We say $e \stricttho{\tr} e'$ if $e \tho{\tr} e'$ but $e \neq e'$. 
	The \emph{reads-from} function $\rf{\tr}$ is a map
	from the read events to the write events in $\tr$.
	Under sequential consistency, for 
	a read event $e$ on variable $x$, we have that $e' = \rf{\tr}(e)$
	be the latest write event on the same variable $x$
	such that $e' \trord{\tr} e$.
	We say that a lock $\lk \in \locks{\tr}$ is held at an event
	$e \in \events{\tr}$ if there is an event $e'$
	such that (i)~$\OpOf{e'} = \acq(\lk)$, 
	(ii)~$e' \stricttho{\tr} e$, and
	(iii)~either $\match{\tr}(e')$ does not exist in $\tr$,
	or $e \tho{\tr} \match{\tr}(e')$.
	We use $\lheld{\tr}(e)$ to denote the set of all the locks
	that are held by $\ThreadOf{e}$ right before $e$.
	The lock nesting depth of $\tr$ is $\max\limits_{e \in \events{\tr}} |\lheld{\tr}(e)|+1$ where $\OpOf{e} = \acq(\lk)$.
}

\myparagraph{Deadlock patterns}{
	A deadlock pattern\footnote{Similar notions have been used in the literature, sometimes under the term \emph{deadlock potential~\cite{Havelund2000}.}} of size $k$ in a trace $\tr$ is a sequence 
	$D = \pattern{e_0, e_1, \ldots, e_{k-1}}$,
	with $k$ distinct threads $t_0, \ldots, t_{k-1}$
	and $k$ distinct locks $\lk_0, \ldots, \lk_{k-1}$ such that
	$\ThreadOf{e_i} = t_i$, $\OpOf{e_i} = \acq(\lk_i)$,
	$\lk_i \in \lheld{\tr}(e_{(i+1)\%k})$,
	and further, $\lheld{\tr}(e_i) \cap \lheld{\tr}(e_j) = \emptyset$ for every 
	$i, j$ such that $i \neq j$ and $0 \leq i, j < k$.
	A deadlock pattern is a necessary but insufficient condition of an actual deadlock,
	due to subtle synchronization or control and data flow in the underlying program.
}

\myparagraph{Dynamic predictive analysis and correct reorderings}{
	Dynamic analyses aim to expose bugs by observing traces $\tr$ of a concurrent program, often without accessing the source code.
	While such purely dynamic approaches enjoy the benefits of scalability,
	simply detecting bugs that manifest on $\tr$ offer poor coverage and are bound to miss bugs that appear in select thread interleavings~\cite{Musuvathi2008}.
	Therefore, for better coverage, \emph{predictive} dynamic techniques are developed. 
	Such techniques predict the occurrence of bugs in alternate executions that can be \emph{inferred} from $\tr$,
	irrespective of the program that produced $\tr$.
	The notion of such inferred executions is formalized by the notion
	of correct reorderings~\cite{Smaragdakis12,serbanuta2013,Koushik05}.
	
	A trace $\rho$ is a \emph{correct reordering} of a trace $\tr$ if 
	\begin{enumerate*}[label=(\arabic*)]
		\item $\events{\rho} \subseteq \events{\tr}$,
		\item for every $e, f \in \events{\tr}$ with $e \tho{\tr} f$,
		if $f \in \events{\rho}$, then $e \in \events{\rho}$
		and $e \tho{\rho} f$, and
		\item for every read event $r \in \events{\rho}$, we have
		$\rf{\tr}(r) \in \events{\rho}$ and
		$\rf{\rho}(r)=\rf{\tr}(r)$.
	\end{enumerate*}
	Intuitively, a correct reordering $\rho$ of $\tr$ is a permutation of $\tr$
	that respects the thread order and preserves the values of each read and write
	that occur in $\rho$.
	This ensures a key property --- every program that generated 
	$\tr$ is also capable of generating $\rho$
	(possibly under a different thread schedule),
	and thus $\rho$ serves as a true witness of a bug.
}

\myparagraph{Predictable deadlocks}{
	We say that an event $e$
	is $\tr$-enabled in a correct reordering $\rho$ of $\tr$
	if $e \in \events{\tr}$, $e \not\in \events{\rho}$
	and for every $f \in \events{\tr}$ if $f \stricttho{\tr} e$,
	then $f \in \events{\rho}$.
	A deadlock pattern $D = \pattern{e_0, e_1 \ldots e_{k-1}}$
	of size $k$ in trace $\tr$ is said to be a predictable deadlock
	if there is a correct reordering $\rho$ of $\tr$
	such that each of $e_0, \ldots, e_{k-1}$
	are $\tr$-enabled in $\rho$.
	This notion guarantees that the witness $\rho$ is a valid execution of \emph{any}
	concurrent program that produced $\tr$.
	Analogous definitions have also been widely used for other predictable bugs~\cite{Huang14,Smaragdakis12}.
	We call a deadlock-prediction algorithm \emph{sound} if for every input trace $\tr$, 
	all deadlock reports on $\tr$ are predictable deadlocks of $\tr$ (i.e., no false positives), 
	and \emph{complete} if all predictable deadlocks of $\tr$ are reported by the algorithm (i.e., no false negatives).
	This is in line with the previous works on the topic of predictive analyses~\cite{Cai2021, Kalhauge2018, Mathur2021, Pavlogiannis2020}.
	We remark that other domains sometimes use this terminology reversed.
}
%
%
%
\begin{example}
\exlabel{prelim-defs}
Let us illustrate these definitions on the trace $\tr_2$
in \figref{motivating-dl}, with $e_i$ denoting the $i^\text{th}$ event in the figure.
The set of events,
threads, variables and locks of $\tr_2$ are respectively $\events{\tr_2} = \set{e_i}_{i=1}^{20}$,
$\threads{\tr_2} = \set{t_1, t_2, t_3, t_4}$,
$\vars{\tr_2} = \set{x, y, z}$ and 
$\locks{\tr_2} = \set{\lk_1, \lk_2, \lk_3}$.
The trace order yields
$e_i \trord{\tr_2} e_j$ iff $i \leq j$,
and some examples of thread-ordered events are
$e_1 \stricttho{\tr_2} e_2 \stricttho{\tr_2} e_{15}$ and $e_{16} \stricttho{\tr_2} e_{18} \stricttho{\tr_2} e_{20}$.
The reads-from function is as follows: 
$\rf{\tr_2}(e_{10}) = e_5$, $\rf{\tr_2}(e_{14}) = e_9$ and $\rf{\tr_2}(e_{17}) = e_{13}$.
The lock nesting depth of $\tr_2$ is $2$.
The sequence $D = \pattern{e_4, e_{18}}$ forms a deadlock pattern
because of the cyclic acquisition of locks $\lk_2$ and $\lk_3$
without simultaneously holding a common lock.
The trace
$\rho_4 = e_3..e_7\,e_8..e_{11}\,e_1 e_2 e_{12}..e_{15}\,e_{16} e_{17}$
is a correct reordering of $\tr_2$;
even though it differs from $\tr_2$ in the relative order of the critical sections of lock $\lk_1$, and
contains only a prefix of thread $t_3$, it is consistent with $\rf{\tr_2}$ and $\stricttho{\tr_2}$.
However, $\rho_4$ does not witness $\pattern{e_4, e_{18}}$ as a deadlock, as
only $e_{18}$ is $\tr_2$-enabled in $\rho_4$.
On the other hand, the trace $\rho_3 = e_1 e_2 e_3 e_8 e_9 \,e_{12}..e_{15}\,e_{16} e_{17}$
is a correct reordering of $\tr_2$ in which $e_4$ and $e_{18}$ are $\tr_2$-enabled,
witnessing $D$ as a predictable deadlock of $\tr_2$. 
\end{example}

\section{The Complexity of Dynamic Deadlock Prediction}
\seclabel{lower-bounds}

Detecting deadlock patterns and predictable deadlocks is clearly a problem in $\NP$,
as any witness for either problem can be verified in polynomial time.
However, little has been known about the hardness of the problem in terms of rigorous lower bounds.
Here we settle these questions, by proving strong intractability results.
Due to space constraints, we state and explain the main results here, and refer to\begin{pldi}~our technical report \cite{arxiv}\end{pldi}\begin{arxiv}~\cref{sec:sec:app_proofs_lower_bounds}\end{arxiv} for the full proofs.

\myparagraph{Parametrized hardness for detecting deadlock patterns}{
We show that the basic problem of checking the existence of a deadlock \emph{pattern} is itself
hard parameterized by the size $k$
of the pattern.

\begin{restatable}{theorem}{patternwonehardness}
\thmlabel{pattern-w1-hardness-pattern}
Checking if a trace $\tr$ contains a deadlock pattern of size $k$
is $\W{1}$-hard in the parameter $k$.
Moreover, the problem remains $\NP$-hard even when the lock-nesting depth of $\tr$ is constant.
\end{restatable}
}

\begin{proof}
	We show that there is a polynomial-time fixed parameter tractable reduction from
	INDEPENDENT-SET(c) to the problem of checking the existence of deadlock-patterns
	of size $c$.
	Our reduction takes as input an undirected graph $G$ and outputs a trace $\tr$
	such that $G$ has an independent set of size $c$ iff $\tr$
	has a deadlock pattern of size $c$.
	
	\begin{figure}[]
		\input{figures/construction-w1-hardness}
		\hfill ~~ 
\begin{subfigure}[b]{0.475\textwidth}
	\newcommand{\xdisposition}{0}
	\newcommand{\ydisposition}{0}
	\newcommand{\xtstep}{0.75}
	\newcommand{\ytstep}{1}
	\newcommand{\ybias}{-0.3 }
	\newcommand{\xstep}{2.5}
	\newcommand{\ystep}{-0.475}
	\newcommand{\xtscale}{0.8}
	\def \numevents{9.5}
	\newcommand{\eventA}[4]{
		\node[event, draw=black, fill=white] (A#1) at (#1*\xstep, #2*\ystep) {\footnotesize $#2(x_{#3})$};
	}
	\scalebox{0.9}{
		\begin{tikzpicture}[thick,
			pre/.style={<-,shorten >= 2pt, shorten <=2pt, very thick},
			post/.style={->,shorten >= 3pt, shorten <=3pt,   thick},
			seqtrace/.style={line width=2},
			und/.style={very thick, draw=gray},
			event/.style={rectangle, minimum height=0.8mm, minimum width=15mm,  line width=1pt, inner sep=0.5,},
			virt/.style={circle,draw=black!50,fill=black!20, opacity=0}]
			\footnotesize


			\begin{scope}[shift={(0,-3*\ystep)}]
				
				\draw[dashed] (-0.9*\xstep,6.5*\ystep) rectangle (-0.5*\xstep,8.5*\ystep);
				\node (A1) at (-0.7*\xstep,7*\ystep) {\normalsize [1, 1]};
				\node (A2) at (-0.7*\xstep,8*\ystep) {\normalsize [1, 0]};
				\node (A) at (-0.7*\xstep,9*\ystep) {\large  $A$};
				
				\draw[dashed] (1.5*\xstep,6.5*\ystep) rectangle (1.9*\xstep,8.5*\ystep);
				\node (B1) at (1.7*\xstep,7*\ystep) {\normalsize [1, 0]};
				\node (B2) at (1.7*\xstep,8*\ystep) {\normalsize [0, 1]};
				\node (B) at (1.7*\xstep,9*\ystep) {\large  $B$};

			\end{scope}

			\node[] (S11) at (0*\xstep,0.15) {\normalsize $t_A$};
			\node[] (S12) at (0*\xstep,\numevents * \ystep) {};
			\node[] (S21) at (1*\xstep,0.15) {\normalsize $t_B$};
			\node[] (S22) at (1*\xstep,\numevents * \ystep) {};
			
			\draw[seqtrace] (S11) to (S12);
			\draw[seqtrace] (S21) to (S22);

			\node[event, draw=black, fill=white] (11) at (0*\xstep, 1*\ystep + 0*\ybias) {$\acq(\LockColorTwo{\lk_2})$};
			\node[event, draw=black, fill=white] (12) at (0*\xstep, 2*\ystep + 0*\ybias) {$\acq(\LockColorOne{\lk_1})$};
			\node[event, draw=black, fill=white, dotted] (13) at (0*\xstep, 3*\ystep + 0*\ybias) {$\cs(m_0,m_1)$};
			\node[event, draw=black, fill=white] (17) at (0*\xstep, 4*\ystep + 0*\ybias) {$\rel(\LockColorOne{\lk_1})$};
			\node[event, draw=black, fill=white] (18) at (0*\xstep, 5*\ystep + 0*\ybias) {$\rel(\LockColorTwo{\lk_2})$};
			
			\node[event, draw=black, fill=white] (19) at (0*\xstep, 6*\ystep + 1*\ybias) {$\acq(\LockColorOne{\lk_1})$};
			\node[event, draw=black, fill=white, dotted] (110) at (0*\xstep, 7*\ystep + 1*\ybias) {$\cs(m_0,m_1)$};
			\node[event, draw=black, fill=white] (114) at (0*\xstep, 8*\ystep + 1*\ybias) {$\rel(\LockColorOne{\lk_1})$};
			
			
			\node[event, draw=black, fill=white] (21) at (1*\xstep, 2*\ystep + 0*\ybias) {$\acq(\LockColorOne{\lk_1})$};
			\node[event, draw=black, fill=white, dotted] (22) at (1*\xstep, 3*\ystep + 0*\ybias) {$\cs(m_1,m_0)$};
			\node[event, draw=black, fill=white] (26) at (1*\xstep, 4*\ystep + 0*\ybias) {$\rel(\LockColorOne{\lk_1})$};
			
			\node[event, draw=black, fill=white] (27) at (1*\xstep, 5*\ystep + 1*\ybias) {$\acq_1(\LockColorTwo{\lk_2})$};
			\node[event, draw=black, fill=white, dotted] (28) at (1*\xstep, 6*\ystep + 1*\ybias) {$\cs(m_1,m)$};
			\node[event, draw=black, fill=white] (212) at (1*\xstep, 7*\ystep + 1*\ybias) {$\rel_1(\LockColorTwo{\lk_2})$};
			
			\begin{scope}[]
				\node[below left=of 212] (212b) {};
			\end{scope}
	
		\end{tikzpicture}
	}
	\caption{
		Reduction for OV-hardness proof from an instance of size $n = 2$ and $d = 2$.
	}
	\figlabel{ov-hardness}
\end{subfigure}
		\caption{
		Construction of $\W{1}$-hardness (\subref{fig:w1-hardness}) and OV-hardness (\subref{fig:ov-hardness}) results. 
		We use the shortcut $\cs(\lk,\lk')$ to denote two nested critical sections on $\lk$ and $\lk'$. That is,
	$\cs(\lk,\lk')=\acq(\lk) \cdot\acq(\lk') \cdot\rel(\lk') \cdot \rel(\lk)$.
		} 
	\end{figure}

	\myparagraph{Construction}{
		Let $V = \set{v_1, v_2, \ldots, v_n}$.
		We assume a total ordering $<_E$ on the set of edges $E$.
		The trace $\tr$ we construct is a concatenation of $c$ sub-traces: 
		$
		\tr = \tr^{(1)} \cdot \tr^{(2)} \cdots \tr^{(c)}
		$
		and uses $c$ threads $\set{t_1, t_2, \ldots t_c}$ and
		$|E| + c$ locks $\set{\lk_{\set{u, v}}}_{\set{u, v} \in E} \uplus \set{\lk_0, \lk_1 \ldots, \lk_{c-1}}$.
		The $i^\text{th}$ sub-trace $\tr^{(i)}$ is a sequence of events performed by thread $t_i$, and
		is obtained by concatenation of $n = |V|$ sub-traces:
		$
		\tr^{(i)} = \tr^{(i)}_1 \cdot \tr^{(i)}_2 \cdots \tr^{(i)}_n
		$.
		Each sub-trace $\tr^{(i)}_j$ with $(i \leq c, j \leq n)$ 
		comprises of nested critical sections over locks of the 
		form $\lk_{\set{v_j, u}}$, where $u$ is a neighbor of $v_j$.
		Inside the nested block we have critical 
		sections on locks $\lk_{i \% c}$ and $\lk_{(i+1) \% c}$.
		Formally, let $\set{v_j, v_{k_1}}, \ldots, \set{v_j, v_{k_d}}$
		be the neighboring edges of $v_j$ (ordered according to $<_E$).
		Then, $\tr^{(i)}_j$ is the unique string generated by the
		grammar having $d+1$ non-terminals $S_0, S_1, \ldots, S_d$, start symbol $S_d$
		and the following production rules:
		\begin{itemize}
			\item $S_0 \to \ev{t_i, \acq(\lk_{i \% c})} \cdot \ev{t_i, \acq(\lk_{(i+1) \% c})} \cdot \ev{t_i, \rel(\lk_{(i+1) \% c})} \cdot \ev{t_i, \rel(\lk_{i \% c})}$.
			\item for each $1 \leq r \leq d$, $S_r \to \ev{t_i, \acq(\lk_{\set{v_j, v_{k_r}}})} \cdot S_{r-1} \cdot \ev{t_i, \rel(\lk_{\set{v_j, v_{k_r}}})}$.
		\end{itemize}
		\figref{w1-hardness} illustrates this construction for a graph with $3$ nodes
		and parameter $c = 3$.
		Finally, observe that the lock-nesting depth in $\tr$ is bounded by 2 + the degree of $G$.
	}
\end{proof}

\thmref{pattern-w1-hardness-pattern} implies that the problem is not only $\NP$-hard, 
but also unlikely to be \emph{fixed parameter tractable} in the size $k$ of the deadlock pattern.
In fact, under the well-believed Exponential Time Hypothesis (ETH), the parametrized
problem INDEPENDENT-SET(c) cannot be solved in time $f(c) \cdot n^{o(c)}$~\cite{chen2006strong}.
The above reduction preserves the parameter $k=c$, thus under ETH,
detecting deadlock patterns of size $k$ is unlikely to be solvable 
in time complexity $f(k) \cdot \NumEvents^{g(k)}$, where 
$g(k)$ is $o(k)$ (such as $g(k) = \sqrt{k}$ or even $g(k) = k/\log(k)$).
The problem of checking the existence of deadlock patterns 
is, intuitively, a precursor to the deadlock prediction problem. 
Thus, an approach for 
deadlock prediction that first identifies the existence of arbitrary deadlock patterns and 
then verifying their feasibility is unlikely to be tractable. 
In practice, the synchronization patterns corresponding to the hard instances are uncommon in executions, and our proposed algorithms (\secref{syncp} and \secref{otf}) can effectively expose predictable deadlocks (\secref{experiments}).

\myparagraph{Fine-grained hardness for deadlock pattern detection}{
We next consider the problem of
detecting deadlock patterns of size $2$, as these form the most common case in practice~\cite{Lu08}.
Observe that \thmref{pattern-w1-hardness-pattern} has no implications on this case, as here $k$ is fixed.
The problem admits a folklore $O(\NumEvents^2)$ time algorithm, by iterating over all pairs of lock-acquisition events of the input trace, and checking whether any such pair forms a deadlock pattern.
Perhaps surprisingly, here we show that, despite its simplicity, this algorithm is optimal, i.e., we cannot hope to improve over this quadratic bound.
This result is based on a reduction from the popular Orthogonal Vectors (OV) problem.
Given two sets of $d$-dimensional vectors
$A, B \subseteq \set{0, 1}^d$
of cardinality $|A| = |B| = n$, the OV problem asks if there are $a \in A, b \in B$
such that $a \cdot b = \sum_i a[i]\cdot b[i]= 0$.
The OV hypothesis states that for any 
$\epsilon >0$, there is no $O(n^{2-\epsilon}\cdot \poly{d})$ algorithm for solving OV.
This is also a consequence of the famous Strong Exponential Time Hypothesis (SETH)~\cite{williams2005}.
We next show that detecting deadlock patterns of size $2$ is at least as hard as solving OV.

\begin{restatable}{theorem}{patternovhardness}
\thmlabel{pattern-ov-hardness}
Given a trace $\tr$ of size $\NumEvents$ and $\NumLocks$ locks, 
for any $\epsilon>0$,
there is no algorithm that determines in $O(\NumEvents^{2-\epsilon}\cdot \poly{\NumLocks})$ time
whether $\tr$ has a deadlock pattern of size $2$, under the OV hypothesis.
\end{restatable}


\begin{proof}
	We show a fine-grained reduction from the Orthogonal Vectors Problem to
	the problem of checking for deadlock patterns of size $2$.
	For this, we start with two sets
	$A, B \subseteq \set{0, 1}^d$
	of $d$-dimensional vectors with $|A| = |B| = n$.
	We write the $i^{th}$ vector in $A$ as $A_i$ and that in $B$ as $B_i$.

	\myparagraph{Construction}{
		We will construct a trace $\tr$ such that $\tr$ has a deadlock
		pattern of length $2$ iff $(A, B)$ is a positive OV instance.
		The trace $\tr$ is of the form $\tr = \tr^A \cdot \tr^B$
		and uses $2$ threads $\set{t_A, t_B}$ and $d+2$ distinct locks $\lk_1, \ldots, \lk_d, m_0, m_1$. 
		Intuitively, $\tr^A$ and $\tr^B$  encode the given sets of vectors $A$ and $B$.
		The sub-traces $\tr^A = \tr^A_1 \cdot \tr^A_2 \cdots \tr^A_n$
		and $\tr^B = \tr^B_1 \cdot \tr^B_2 \cdots \tr^B_n$ are defined as follows.
		For each $i \in \set{1, 2, \ldots, n}$ and $Z \in \set{A, B}$, the
		sub-trace $\tr^Z_i$ is the unique string generated by the
		grammar having $d+1$ non-terminals $S_0, S_1, \ldots, S_d$, start symbol $S_d$
		and the following production rules:
		\begin{itemize}
			\item $S_0 \to \ev{t_Z, \acq(m)} \cdot \ev{t_Z, \acq(m')} \cdot \ev{t_Z, \rel(m')} \cdot \ev{t_Z, \rel(m)}$,
			where $(m, m') = (m_0, m_1)$ if $Z = A$, and $(m, m') = (m_1, m_0)$.
			\item for each $1 \leq j \leq d$, $S_j \to S_{j-1}$ if $Z_i[j] = 0$.
			Otherwise (if $Z_i[j] =1$), $S_j \to \ev{t_Z, \acq(\lk_j)} \cdot S_{j-1} \cdot \ev{t_Z, \rel(\lk_j)}$.
		\end{itemize}
		In words, all events of $\tr^A$ are performed by thread $t_A$ and those
		in $\tr^B$ are performed by $t_B$.
		Next, the $i^{th}$ sub-trace of $\tr^A$, denoted $\tr^A_i$ corresponds to the vector $A_i$
		as follows --- $\tr^A_i$ is a nested block of critical sections,
		with the innermost critical section being on lock $\lk'$,
		which is immediately enclosed in a  critical section on lock $\lk$.
		Further, in the sub-trace $\tr^A_i$, the lock $\lk_j$
		occurs iff $A_i[j] = 1$.
		The sub-traces $\tr^B_i$ is similarly constructed, except that the order
		of the two innermost critical sections is inverted.
		\figref{ov-hardness} illustrates the construction for an OV-instance with $n=2$ and $d=2$.}
\end{proof}

\myparagraph{The complexity of deadlock prediction}{
Finally, we settle the complexity of the prediction problem for deadlocks,
and show that, even for deadlock patterns of size $2$, the problem is $\W{1}$-hard
parameterized by the number of threads.
In contrast, recall that the $\W{1}$-hardness of 
\thmref{pattern-w1-hardness-pattern} concerns deadlock patterns of arbitrary size.
Our result is based on a similar hardness that was established recently for predicting data races~\cite{Mathur2020b}.

\begin{restatable}{theorem}{wonehardness}
\thmlabel{w1-hardness-pattern}
The problem of checking if a trace $\tr$ has a predictable deadlock of size $2$
is $\W{1}$-hard in the number of threads $\NumThreads$ appearing in $\tr$, and thus
is also $\NP$-hard.
\end{restatable}
}

\section{Synchronization-Preserving Deadlocks and their Prediction}
\seclabel{syncp}
Having established the intractability of general deadlock prediction in \secref{lower-bounds}, we now define the subclass of predictable deadlocks called synchronization-preserving (\emph{sync-preserving}, for short) in \secref{syncp-def}.
The key benefit of sync-preserving deadlocks is that, unlike arbitrary deadlocks, they can be detected efficiently; we develop our algorithm $\SyncPDOffline$ for this task in Sections 4.2-4.5.
Our experiments later indicate that most predictable deadlocks are actually sync-preserving, hence the benefit of fast detection comes at the cost of little-to-no precision loss in practice.

\myparagraph{Overview of the algorithm}{
	There are several insights behind our algorithm.
	First, given a deadlock pattern, one can verify 
	if it is a sync-preserving deadlock in linear time (\secref{verify-patterns});
	this is based on our sound and complete characterization of sync-preserving deadlocks (\secref{characterize-patterns}).
	Next, instead of verifying single deadlock patterns one-by-one, 
	we consider \emph{abstract deadlock patterns},
	which are essentially collections of deadlock patterns
	that share the same signature;
	the formal definition is given in \secref{verify-abstract-patterns}.
	We show that our basic algorithm can
	be extended to \emph{incrementally}
	verify \emph{all} the concretizations of an abstract deadlock pattern in linear time (\secref{verify-abstract-patterns}), in a single pass (\lemref{abstract-pattern-linear-time}).
	Finally, we feed this algorithm all the abstract deadlock patterns of the input trace,
    by constructing an \emph{abstract lock graph} and enumerating cycles in it (\secref{enumerate-patterns}).
	
}

\subsection{Synchronization-Preserving Deadlocks}
\seclabel{syncp-def}

Our notion of sync-preserving deadlocks builds on the recently introduced concept of sync-preserving correct reorderings~\cite{Mathur2021}.

\begin{definition}[Sync-preserving Correct Reordering]
\deflabel{syncp-correct-reordering}
A correct reordering $\rho$ of a trace $\tr$ is
\emph{sync-preserving} if for every lock $\lk \in \locks{\rho}$
and every two acquire events $e_1 \neq e_2 \in \events{\rho}$
with $\OpOf{e_1} = \OpOf{e_2} = \acq(\lk)$, the order of $e_1$ and
$e_2$ is the same in $\tr$ and $\rho$, i.e., 
$e_1 \trord{\rho} e_2$ iff $e_1 \trord{\tr} e_2$.
\end{definition} 
A sync-preserving correct reordering
preserves the order of those critical sections (on the same lock) that actually appear in 
the reordering, but allows 
intermediate critical sections to be dropped completely.
%
This style of reasoning is more permissive than
the space of reorderings 
induced by the Happens-Before partial order~\cite{Lamport78},
that implicitly enforces that all intermediate critical sections on a lock
be present.
Sync-preserving deadlocks can now be defined naturally.
\begin{definition}[Sync-preserving Deadlocks]
\deflabel{syncp-deadlock}
Let $\tr$ be a trace and $D = \pattern{e_0, e_1, \ldots, e_{k-1}}$
be a deadlock pattern.
We say that $D$ is a sync-preserving deadlock of $\tr$
if there is a sync-preserving correct reordering $\rho$ of $\tr$
such that each of $e_0, \ldots, e_{k-1}$ is $\tr$-enabled in $\rho$.
\end{definition}


\begin{figure}
	\centering
	\begin{subfigure}[t]{0.17\textwidth}
		\scalebox{0.9}{
			\execution{3}{
				\figev{1}{$\acq(\LockColorOne{\lk_1})$}
				\figev{1}{$\acq(\LockColorTwo{\lk_2})$}
				\figev{1}{$\rel(\LockColorTwo{\lk_2})$}
				\figev{1}{$\acq(\LockColorTwo{\lk_2})$}
				\figev{1}{$\wt(y)$}
				\figev{1}{$\rel(\LockColorTwo{\lk_2})$}
				\figev{1}{$\rel(\LockColorOne{\lk_1})$}
				\figev{2}{$\acq(\LockColorThree{\lk_3})$}
				\figev{2}{$\wt(x)$}
				\figev{2}{$\rd(y)$}
				\figev{2}{$\rel(\LockColorThree{\lk_3})$}
				\figev{3}{$\acq(\LockColorTwo{\lk_2})$}
				\figev{3}{$\acq(\LockColorThree{\lk_3})$}
				\figev{3}{$\rd(x)$}
				\figev{3}{$\rel(\LockColorThree{\lk_3})$}
				\figev{3}{$\acq(\LockColorOne{\lk_1})$}
			}
		}
	\end{subfigure}
	\hspace*{\fill}
	\begin{subfigure}[t]{0.17\textwidth}
		\scalebox{0.9}{
			\execution{3}{
				\figevoffset{16}{3}{$\wt(v)$}
				\figevoffset{16}{3}{$\rel(\LockColorOne{\lk_1})$}
				\figevoffset{16}{3}{$\acq(\LockColorOne{\lk_1})$}
				\figevoffset{16}{3}{$\rel(\LockColorOne{\lk_1})$}
				\figevoffset{16}{3}{$\rel(\LockColorTwo{\lk_2})$}
				\figevoffset{16}{2}{$\acq(\LockColorFour{\lk_4})$}
				\figevoffset{16}{2}{$\acq(\LockColorOne{\lk_1})$}
				\figevoffset{16}{2}{$\wt(z)$}
				\figevoffset{16}{2}{$\rd(v)$}
				\figevoffset{16}{2}{$\rel(\LockColorOne{\lk_1})$}
				\figevoffset{16}{2}{$\rel(\LockColorFour{\lk_4})$}
				\figevoffset{16}{1}{$\acq(\LockColorOne{\lk_1})$}
				\figevoffset{16}{1}{$\acq(\LockColorTwo{\lk_2})$}
				\figevoffset{16}{1}{$\rd(z)$}
				\figevoffset{16}{1}{$\rel(\LockColorTwo{\lk_2})$}
				\figevoffset{16}{1}{$\rel(\LockColorOne{\lk_1})$}
			}
		}
	\end{subfigure}
	\hspace*{\fill}
	\begin{subfigure}[t]{0.35\textwidth}
		\begin{tikzpicture}[baseline=-180pt]
			\tikzset{rectangle/.append style={draw=black}}
			\scalebox{0.84}{
				\def\xstep{2.2}
				\def\ystep{0.5}
				\node[align=left] (tt1) at (0*\xstep, -3.5){
					$
					\eta_1=
					\tuple{t_1, \LockColorTwo{\lk_2}, \{ \LockColorOne{\lk_1} \},
						\sequence{e_2, e_4, e_{29}}}
					$\\
					$
					\eta_2=
					\tuple{
						t_2, \LockColorOne{\lk_1}, \{ \LockColorFour{\lk_4} \},
						\sequence{e_{23}}}
					$\\
					$
					\eta_3=
					\tuple{
						t_3, \LockColorOne{\lk_1}, \{ \LockColorTwo{\lk_2} \},
						\sequence{e_{16}, e_{19}}}
					$\\
					$
					\eta_4=
					\tuple{
						t_3, \LockColorThree{\lk_3}, \{ \LockColorTwo{\lk_2} \},
						\sequence{e_{13}}}
					$
				};
				
				\node[align=left] (tt2) at (0*\xstep, -4.8){
					$\abst{D}=\pattern{\eta_1, \eta_3}$
				};
				
				\node[align=left] (tt3) at (-0.5*\xstep, -6){
					$D_1=\pattern{e_2, e_{16}}$\\
					$D_2=\pattern{e_2, e_{19}}$\\
					$D_3=\pattern{e_4, e_{16}}$\\
				};
				\node[align=left] (tt4) at (0.5*\xstep, -6){
					$D_4=\pattern{e_4, e_{19}}$\\
					$D_5=\pattern{e_{29}, e_{16}}$\\
					$D_6=\pattern{e_{29}, e_{19}}$\\
				};
			}
		\end{tikzpicture}
	\end{subfigure}
	\caption{
		A trace $\tr_3$, its abstract acquires $\eta_i$, unique abstract deadlock pattern $\abst{D}$, concrete patterns $D_i\in \abst{D}$. 
	}
	\figlabel{syncp_example}
\end{figure}

\begin{example}
\exlabel{sp-deadlocks}
Consider the trace $\tr_2$ in \figref{motivating-dl}.
The deadlock pattern $D = \pattern{e_4, e_{18}}$ is a sync-preserving deadlock, witnessed by the sync-preserving correct reordering
$\rho_3 = e_1e_2e_3 e_8 e_9 \,e_{12}..e_{15}\,e_{16} e_{17}$.
Now consider the trace $\tr_3$ from \figref{syncp_example}
and the deadlock pattern $D_5 = \pattern{e_{29}, e_{16}}$.
This is a predictable deadlock, witnessed by the correct reordering
$\rho_5 = e_1..e_7e_8..e_{11}e_{12}..e_{15}\,e_{28}$.
Observe that $\rho_5$ is a sync-preserving reordering, which makes $D_5$ a sync-preserving deadlock.
A key aspect in $\rho_5$ is that the events $e_{22}..e_{27}$ are dropped, as otherwise $e_{16}$ cannot be $\sigma_3$-enabled.
A similar reasoning applies for the deadlock pattern $D_6$, and it is also a sync-preserving deadlock.
The other deadlock patterns ($D_1, D_2, D_3, D_4$)
are not predictable deadlocks.
Intuitively, the reason for this is that realizing these deadlock patterns require executing the read event $e_{14}$, which then enforces to execute the events $e_8..e_{11}$ and $e_1..e_6$.
This prevents the deadlocks from becoming realizable as the events 
$e_2$ or $e_4$ that appear in these deadlock patterns are no longer $\sigma_3$-enabled.
This point is detailed in \exref{syncp-ex}.
\end{example}


\subsection{Characterizing Sync-Preserving Deadlocks}
\seclabel{characterize-patterns}

There are two fundamental tasks in 
searching for a correct reordering that witnesses a deadlock --- 
(i)~determining the set of events in the correct reordering, and 
(ii)~identifying a total order on such events --- both of which are intractable~\cite{Mathur2020b}.
On the contrary, for sync-preserving deadlocks, we show that
(a)~the search for a correct reordering can be reduced
to the problem of checking if some well-defined set of events (\defref{spclosure})
does not contain the events appearing in the deadlock pattern (\lemref{spclosure-deadlock}),
and that (b)~this set can be constructed efficiently.
\begin{definition}[Sync-Preserving Closure]
\deflabel{spclosure}
Let $\tr$ be a trace and $S \subseteq \events{\tr}$.
The sync-preserving closure of $S$, denoted
$\SPClosure{\tr}(S)$ is the smallest set $S'$ such that
\begin{enumerate*}[label=(\alph*)]
	\item \itmlabel{def-item-spclosure-a} $S \subseteq S'$,
	\item \itmlabel{def-item-spclosure-b} for every $e, e' \in \events{\tr}$ such that 
	$e \stricttho{\tr} e'$ or $e = \rf{\tr}(e')$, if $e' \in S'$, then $e \in S'$, and
	\item\itmlabel{acq-order-closure} for every lock $\lk$ and every two distinct events $e, e'\in S'$ with $\OpOf{e} = \OpOf{e'} = \acq(\lk)$,
	if $e \trord{\tr} e'$ then $\match{\tr}(e) \in S'$.
\end{enumerate*}
\end{definition}

\defref{spclosure} resembles 
the notion of correct reorderings (\defref{syncp-correct-reordering}). 
Indeed, \lemref{spclosure-necessary-sufficient} justifies using this set --- it is
both a necessary and a sufficient set for sync-preserving correct reorderings.

\begin{restatable}{lemma}{spclosureNecessarySufficient}
\lemlabel{spclosure-necessary-sufficient}
Let $\tr$ be a trace and let $S \subseteq \events{\tr}$.
For any sync-preserving correct reordering $\rho$ of $\tr$, 
	if $S \subseteq \events{\rho}$, then $\SPClosure{\tr}(S) \subseteq \events{\rho}$.
	Further, there is a sync-preserving correct reordering $\rho$ of 
	$\tr$ such that $\events{\rho} = \SPClosure{\tr}(S)$.
\end{restatable}

For an intuition, consider again \cref{fig:syncp_example} and the sync-preserving correct reordering $\rho_5 = e_1..e_7e_8..e_{11}e_{12}..e_{15}\,e_{28}$ computed in \cref{ex:sp-deadlocks}.
According to \cref{lem:spclosure-necessary-sufficient}, $\SPClosure{\tr_3}(S) \subseteq \events{\rho_5}$ holds for all $S$ such that $S \subseteq \events{\rho_5}$.
For example, if we take $S=\set{e_1, e_{15}}$ then observe that $S \subseteq \events{\rho_5} $ and  $\SPClosure{\tr_3}(S)=\set{e_1, \ldots, e_6, \, e_8, \ldots e_{15}} \subseteq \events{\rho_5}$ holds.

Based on~\lemref{spclosure-necessary-sufficient}, we 
present a sound and complete characterization of sync-preserving deadlocks
(\lemref{spclosure-deadlock}).
For a set $S \subseteq \events{\tr}$, 
we let $\prev{\tr}(S)$ denote the set of immediate thread
predecessors of events in $S$. 
That is,
$\prev{\tr}(S) = \setpred{e \in \events{\tr}}{\exists f \in S, e \stricttrord{\tr} f \text{ and } \forall e' \stricttrord{\tr} f, e' \tho{\tr} e}$.

\begin{restatable}{lemma}{spclosureDeadlock}
\lemlabel{spclosure-deadlock}
Let $\tr$ be a trace and let $D = \pattern{e_0, \ldots, e_{k-1}}$
be a deadlock pattern of size $k$ in $\tr$.
$D$ is a sync-preserving deadlock of $\tr$ iff
$\SPClosure{\tr}(\prev{\tr}(S)) \cap S = \emptyset$, where
$S = \set{e_0, \ldots, e_{k-1}}$.
\end{restatable}


\begin{example}
\exlabel{syncp-ex}
%
Consider the trace $\tr_2$ in \figref{motivating-dl}, 
and the deadlock pattern $D=\pattern{e_{4}, e_{18}}$.
We have 
$\SPClosure{\tr_2}(\prev{\tr_2}(\{ e_{4}, e_{18} \})) = 
\set{e_1, e_2, e_3, e_8, e_9, e_{12}, \ldots, e_{17}}$.
Since we have that $e_{4}, e_{18} \not \in$ $\SPClosure{\tr_2}( \prev{\tr_2}(\set{e_{4},  e_{18}}))$, $D$ is a sync-preserving deadlock. 
Now consider the trace $\tr_3$ in \figref{syncp_example},
and the deadlock patterns
$D_1=\allowbreak\pattern{e_2, e_{16}}$, 
$D_5=\allowbreak\pattern{e_{29}, e_{16}}$, and
$D_6=\allowbreak\pattern{e_{29}, e_{19}}$.
We have
$\SPClosure{\tr_3}(\prev{\tr_3}(\{ e_{2},e_{16}\}))\allowbreak=
\set{e_1, \ldots, e_6, e_8, \ldots, e_{15}}$, 
$\SPClosure{\tr_3}(\prev{\tr_3}(\allowbreak\{ e_{29},e_{16}\}))=
\set{
e_1, \ldots, e_{15}, e_{28}}$, and $\SPClosure{\tr_3}(\prev{\tr_3}(\{ e_{29},e_{19}\}))=
\set{
e_1, \ldots, e_{18}, e_{28}}$.
Since $e_2 \in \allowbreak \SPClosure{\tr_3}(\allowbreak\prev{\tr_3}(\{ e_{2},e_{16}\}))$,
$D_1$ is not a sync-preserving deadlock.
However, $e_{29}, e_{16} \not \in \SPClosure{\tr_3}(\prev{\tr_3}(\{ e_{29},e_{16}\}))$, and $e_{29}, e_{19} \not \in \SPClosure{\tr_3}(\prev{\tr_3}(\{ e_{29},e_{19}\}))$, thus $D_5$ and $D_6$ are sync-preserving deadlocks (as we also concluded in \exref{sp-deadlocks}).
\end{example}


\subsection{Verifying Deadlock Patterns}
\seclabel{verify-patterns}

Given a deadlock pattern, we check if it constitutes a sync-preserving deadlock
by constructing the sync-preserving closure (\lemref{spclosure-deadlock}) in linear time.
Based on \defref{spclosure}, this can be done in 
an iterative manner. 
We 
(i)~start with the set of $\tho{}$ predecessors of the events in the deadlock pattern, and
(ii)~iteratively  add
$\tho{}$ and $\rf{}$ predecessors of the current set of events.
Additionally, we identify and add the release events that must be included in the set.
%
We utilize \emph{timestamps} to ensure that the \emph{entire} fixpoint computation is
performed in linear time.

\newcommand{\TS}[1]{\mathsf{TS}_{\tr}}

\myparagraph{Thread-read-from timestamps}{
	Given a set $\threads{}$ of threads, a timestamp 
	is simply a mapping $T : \threads{} \to \nats$.
	Given timestamps $T_1, T_2$, we 
	use the notations 
	$T_1 \cle T_2$ and $T_1 \mx T_2$ for pointwise comparison
	and pointwise maximum, respectively.
	For a set $U$ of timestamps, we write $\bigsqcup U$ to denote the
	pointwise maximum over all elements of $U$.
	Let $\trf{\tr}$ be the reflexive transitive closure
	of the relation $(\tho{\tr} \cup \setpred{(\rf{\tr}(e), e)}{\exists x \in \vars{\tr}, \OpOf{e} = \rd(x)})$;
	observe that
	$\trf{\tr}$ is a partial order.
	We define the 
	 timestamp $\TS{\tr}^e$ of an event $e$ in $\tr$
	to be a $\threads{\tr}$-indexed timestamp as follows: $\TS{\tr}^e(t) = |\setpred{f}{f \trf{\tr} e}|$.
	This ensures that
	for two events $e, e' \in \events{\tr}$,
	$e \trf{\tr} e'$ iff $\TS{\tr}^e \cle \TS{\tr}^{e'}$.
	For a set $S \subseteq \events{\tr}$, we overload the notation
	and say the timestamp of $S$ is $\TS{\tr}^S = \bigsqcup \set{\TS{\tr}^e}_{e \in S}$.
	Given a trace $\tr$ with $\NumEvents$ events and
	$\NumThreads$ threads we can compute these timestamps for all the events in $O(\NumEvents \cdot \NumThreads)$ time, using a simple vector clock algorithm~\cite{Mattern89,Fidge91}.
}

\myparagraph{Computing sync-preserving closures}{
	Recall the basic template of the fixpoint computation.
	In each iteration, we identify the set of release events
	that must be included in the set, together with their $\trf{\tr}$-closure.
	In order to identify such events efficiently, for every
	thread $t$ and lock $\lk$, we maintain
	a FIFO queue $\AcqLst_{t, \lk}$ (\emph{critical section history} of $t$ and $\lk$) 
	to store 
	the sequence of events that acquire $\lk$ in thread $t$.
	In each iteration, we traverse each list
	to determine the last acquire event that belongs to the current set.
	For a given lock, we need to add the matching release events of
	all thus identified events to the closure, 
	except possibly the matching release event of the latest acquire event (see~\defref{spclosure}).
	This computation is performed using timestamps, as shown in \algoref{compute-closure}.
	Starting with a set $S$, the algorithm
	runs in time $O(|S|\cdot\NumThreads + \NumThreads\cdot\NumAcquires)$,
	where $\NumThreads$ and $\NumAcquires$ are respectively
	the number of threads and acquire events in $\tr$.
}

\begin{minipage}{0.45\textwidth}



\small
\begin{algorithm}[H]
\Input{Trace $\tr$, Timestamp $T_0$}
\BlankLine
	\Let $\set{\AcqLst_{t,\lk}}_{\lk\in \locks{\tr}, t \in \threads{\tr}}$ be the lock-acquisition histories in $\tr$ \;
	$T \gets T_0$ \;
	\Repeat{$T$ does not change}{
		\For{$\lk \in \locks{}$}{
			\ForEach{$t \in \threads{}$}{
				\Let $e_t$ be the last event in $\AcqLst_{t,\lk}$ with $\TS{\tr}^{e_t} \cle T$ \;
				Remove all earlier events in $\AcqLst_{t, \lk}$ 
			}
			\Let $e_{t*}$ be the last event in $\set{e_t}_{t\in \threads{\tr}}$ according to $\trord{\tr}$ \;
			$T$ := $T \mx \bigsqcup \setpred{\TS{\tr}^{\match{\tr}(e_t)}}{e_t \neq e_{t*}}$	\;
		}
	}
	\Return $T$
\caption{CompSPClosure:\\ Computing sync-preserving closure.}
\algolabel{compute-closure}
\end{algorithm}
\normalsize

\end{minipage}
\hfill ~ 
\begin{minipage}{0.48\textwidth}

\small
\begin{algorithm}[H]
\Input{Trace $\tr$, $\abst{D}$ of length $k$}
\BlankLine
	\Let $F_0, \ldots, F_{k-1}$ be the sequences of acquires in $\abst{D}$ \;
	\Let $n_0, \ldots, n_{k-1}$ be the lengths of $F_0, \ldots, F_{k-1}$ \;
	\lForEach{$j \in \set{0, \ldots, k-1}$}{
		$i_j \gets 1$
	}
	$T \gets \lambda t, 0$ 
	\While{$\bigwedge\limits_{j=0}^{k-1} i_j < n_j$}{
		\Let $e_0 = F_0[i_0], \ldots, e_{k-1} = F_{k-1}[i_{k-1}]$ \;
		$S \gets \prev{\tr}{\set{e_0, \ldots, e_{k-1}}}$\;
		$T \gets$ \fixpoint{$\tr, T \mx \TS{\tr}^S$}\\   \linelabel{call-comp-closure}
		\If{$\forall j < k, \TS{\tr}^{e_{j}} \cle T$}{
			\report pattern $D = e_0, \ldots, e_{k-1}$ and \exit
		}
		\ForEach{$j \in \set{0, \ldots, k-1}$}{
			$i_j = \min\setpred{l \leq n_j}{\TS{\tr}^{F_j[l]} \not\cle T}$ 
		}
	}
\caption{CheckAbsDdlck:\\ Checking an abstract deadlock pattern.}
\algolabel{abstract-pattern}
\end{algorithm}
\normalsize

\end{minipage}

\myparagraph{Checking a deadlock pattern}{
	After computing the timestamp $T$ of
	the closure (output of \algoref{compute-closure}, starting with the set of events in the given deadlock pattern),
	determining whether a given deadlock pattern
	$D = e_0, \ldots, e_{k-1}$ is a sync-preserving deadlock 
	can be performed in time $O(k\cdot\NumThreads)$ --- simply check if  $\forall i, \TS{\tr}(e_i) \not\cle T$.
	This gives an algorithm for checking if a deadlock pattern
	of length $k$ is sync-preserving that runs in time 
	$O(\NumThreads\cdot \NumEvents + k\cdot\NumThreads + \NumThreads\cdot \NumAcquires) = O(\NumEvents\cdot \NumThreads)$.
}

\subsection{Verifying Abstract Deadlock Patterns}
\seclabel{verify-abstract-patterns}

\myparagraph{Abstract acquires and abstract deadlock patterns}{
Given a thread $t$, a lock $\lk$ and a set of locks 
$L\subseteq \locks{\tr}\neq \emptyset$ with $\lk\not \in L$,
we define the \emph{abstract acquire} $\eta =\tuple{t, \lk, L, F}$, 
where $F=\sequence{e_1,\dots, e_n}$ is the sequence of all events
$e_i\in \events{\tr}$ (in trace-order) such that for each $i$, we have 
(i)~$\ThreadOf{e_i}=t$,
(ii)~$\OpOf{e_i} = \acq(\lk)$, and
(iii)~$\lheld{\tr}(e_i) = L$.
In words, the abstract acquire $\eta$ contains the sequence of all acquire events 
of a specific thread, that access a specific lock and hold the same set of locks when executed, ordered as per thread order.
An \emph{abstract deadlock pattern} of size $k$ in a trace $\tr$ is a sequence
$
\abst{D} = \eta_0,\dots,\eta_{k-1}
$
of abstract acquires $\eta_i=\tuple{t_i, \lk_i, L_i, F_i}$ such that
$t_0, \ldots, t_{k-1}$ are distinct threads, 
$\lk_0,  \ldots, \lk_{k-1}$ are distinct locks,
and $L_0, L_1, \ldots, L_{k-1} \subseteq \locks{\tr}$ are
such that $\lk_i \not\in L_i$, $\lk_i \in L_{(i+1)\% k}$ for every $i$,
and $L_i \cap L_j= \emptyset$ for every $i\neq j$.
Thus, an abstract deadlock pattern $\abst{D}$ succinctly encodes all 
concrete deadlock patterns $F_0\times F_1\times \dots \times F_{k-1}$,
called \emph{instantiations} of $\abst{D}$.
We also write $D\in \abst{D}$ to denote that 
$D\in F_0\times F_1\times \dots \times F_{k-1}$.
We say that $\abst{D}$ contains a sync-preserving deadlock if there 
exists some instantiation $D\in \abst{D}$ that is a sync-preserving deadlock.
See \figref{syncp_example} for an example.
}
Our next result is stated below, followed by its proof idea.
\begin{restatable}{lemma}{abstractlineartime}
\lemlabel{abstract-pattern-linear-time}
Consider a trace $\tr$ with $\NumEvents$ events and $\NumThreads$ threads, 
and an abstract deadlock pattern $\abst{D}$ of $\tr$.
We can determine if $\abst{D}$ contains a sync-preserving deadlock in $O(\NumThreads\cdot \NumEvents)$ time.
\end{restatable}



An abstract deadlock pattern of length $k \geq 2$ can have $\NumEvents^k$ instantiations,
giving a naive enumerate-and-check algorithm running 
in time $O(\NumThreads\cdot\NumEvents^{k+1})$, which is prohibitively large.
Instead, we exploit 
(i)~the monotonicity properties of 
the sync-preserving closure (\propref{spclosure-monotone}) and
(ii)~instantiations of an abstract pattern (\corref{pattern-monotone}) 
that allow for an \emph{incremental}
algorithm that iteratively checks successive instantiations
of a given abstract deadlock pattern, while spending total $O(\NumEvents\cdot \NumThreads)$ time.
The first observation allows us to re-use a prior
computation when checking later deadlock patterns.

\begin{restatable}{proposition}{spclosureMonotone}
	\proplabel{spclosure-monotone}
	For a trace $\tr$ and sets $S, S' \subseteq \events{\tr}$.
	If for every event $e \in S$, there is an event $e' \in S'$
	such that $e \tho{\tr} e'$, then
	$\SPClosure{\tr}(S) \subseteq \SPClosure{\tr}(S')$.
\end{restatable}


Consider $\tr_3$ in~\cref{fig:syncp_example} and let $S=\prev{\tr_3}{(\set{e_{29}, e_{16}})}$, and $S'=\prev{\tr_3}{(\set{e_{29}, e_{19}})}$.
The sets $S, S'$ satisfy the conditions of~\cref{prop:spclosure-monotone}, hence $\SPClosure{\tr_3}(S) \subseteq \SPClosure{\tr_3}(S')$, as computed in~\cref{ex:syncp-ex}.
Next, we extend \propref{spclosure-monotone} to avoid redundant
computations when a sync-preserving deadlock is not found and later deadlock patterns must be checked.
Given two deadlock patterns
$D_1 = e_0, \ldots, e_{k-1}$ and $D_2 = f_0, \ldots, f_{k-1}$
of the same length $k$, 
we say $D_1 \prec D_2$ if
they are instantiations of a common abstract pattern $\abst{D}$
(i.e., $D_1, D_2 \in \abst{D}$) and 
for every $i<k$, $e_i \tho{\tr} f_i$.
\begin{restatable}{corollary}{pattern-monotone}
	\corlabel{pattern-monotone}
	Let $\tr$ be a trace and let $D_1 = e_0, \ldots, e_{k-1}$ and 
	$D_2 = f_0, \ldots f_{k-1}$ be deadlock
	patterns of size $k$ in $\tr$ such that $D_1 \prec D_2$.
	Let $S_1 = \set{e_0, \ldots, e_{k-1}}$ and $S_2 = \set{f_0, \ldots, f_{k-1}}$.
	If $\SPClosure{\tr}(\prev{\tr}(S_1))\cap S_2 \neq \emptyset$,
	then $\SPClosure{\tr}(\prev{\tr}(S_2))\cap S_2 \neq \emptyset$.
\end{restatable}

We now describe how \cref{prop:spclosure-monotone} and \cref{cor:pattern-monotone} are used in our algorithms, 
and illustrate them later in  \cref{ex:ex4}. 
\algoref{abstract-pattern} checks if an abstract deadlock 
pattern contains a sync-preserving deadlock.
The algorithm iterates over the sequences $F_0, \ldots, F_{k-1}$ of 
acquires (one for each abstract acquire) in trace order.
For this, it maintains
indices $i_0, \ldots, i_{k-1}$ that point to entries in $F_0, \ldots, F_{k-1}$.
At each step, it determines whether the current deadlock pattern $D = e_0, \ldots, e_{k-1}$
constitutes a sync-preserving deadlock by computing the sync-preserving closure
of the thread-local predecessors of the events of the deadlock pattern.
The algorithm reports a deadlock if the sync-preserving closure does not contain any of $e_0, \ldots, e_{k-1}$.
Otherwise, it looks for the next eligible deadlock pattern,
which it determines based on \corref{pattern-monotone}.
In particular, it advances the pointer $i_j$
all the way until an entry which is outside of the
closure computed so far.
Observe that the timestamp $T$ of the closure
computed in an iteration is being used in later iterations;
this is a consequence of \propref{spclosure-monotone}.
Furthermore, in the call to the \algoref{compute-closure} at \lineref{call-comp-closure}, we ensure that the list of acquires $\AcqLst_{t, \lk}$,
used in the function \fixpoint
is reused across iterations, and not
re-assigned to the original list of all acquire events.
The correctness of this optimization follows from~\propref{spclosure-monotone}.
Let us now calculate the running time of \algoref{abstract-pattern}.
Each of the $\AcqLst_{t, \lk}$ in \fixpoint is traversed
at most once. 
Next, each element of the sequences $F_0, \ldots, F_{k-1}$
is also traversed at most once.
For each of these acquires, the algorithm spends $O(\NumThreads)$
time for vector clock updates.
The total time required is thus $O(\NumEvents \cdot \NumThreads)$.
This concludes the proof of \lemref{abstract-pattern-linear-time}.


\subsection{The Algorithm \SyncPDOffline}
\seclabel{enumerate-patterns}

We now present the final ingredients of \SyncPDOffline.
We construct the \emph{abstract lock graph},
enumerate cycles in it, check whether
any cycle is an abstract deadlock pattern,
and if so, whether it contains sync-preserving deadlocks.


\setlength{\textfloatsep}{10pt}
\begin{figure}[t]
\begin{subfigure}[t]{0.2\textwidth}
\scalebox{0.84}{
\begin{tikzpicture}[yscale=0.9]
	\tikzset{rectangle/.append style={draw=black}}
		\def\xstep{2.2}
		\def\ystep{0.5}
		\node[rectangle, align=center] (t1) at (-1, -7){
			$t_1, \LockColorTwo{\lk_2}, \{ \LockColorOne{\lk_1}\}, \sequence{e_{2}}$
		};
		\node[rectangle, align=center] (t31) at (-1, -8.7){
			$t_2, \LockColorOne{\lk_1}, \{ \LockColorTwo{\lk_2} \}, \sequence{e_{8}}$
		};
		\draw[->, thick, -Latex, bend right=10] ([xshift=-3mm]t1.south) to ([xshift=-3mm]t31.north);
		\draw[->, thick, -Latex, bend right=10] (t31) to (t1);
\end{tikzpicture}
}
\end{subfigure}
~~~~
\begin{subfigure}[t]{0.2\textwidth}
\scalebox{0.84}{
\begin{tikzpicture}[yscale=0.9]
	\tikzset{rectangle/.append style={draw=black}}
		\def\xstep{2.2}
		\def\ystep{0.5}
		\node[rectangle, align=center] (t1) at (-1, -7){
			$t_2, \LockColorThree{\lk_3}, \{ \LockColorTwo{\lk_2}\}, \sequence{e_{4}}$
		};
		\node[rectangle, align=center] (t31) at (-1, -8.7){
			$t_3, \LockColorTwo{\lk_2}, \{ \LockColorThree{\lk_3} \}, \sequence{e_{18}}$
		};
		\draw[->, thick, -Latex, bend right=10] ([xshift=-3mm]t1.south) to ([xshift=-3mm]t31.north);
		\draw[->, thick, -Latex, bend right=10] (t31) to (t1);
\end{tikzpicture}
}
\end{subfigure}
~~~~
\begin{subfigure}[t]{0.4\textwidth}
\scalebox{0.84}{
\begin{tikzpicture}[yscale=0.9]
	\tikzset{rectangle/.append style={draw=black}}
		\node[rectangle, align=center] (t1) at (-1, -7){
					$t_1, \LockColorTwo{\lk_2}, \{ \LockColorOne{\lk_1} \}$,
					$\sequence{e_2, e_4, e_{29}}$
				};
				\node[rectangle,align=center] (t2) at (2.5, -7){
					$t_2, \LockColorOne{\lk_1}, \{ \LockColorFour{\lk_4} \}$,
					$\sequence{e_{23}}$
				};
				\node[rectangle, align=center] (t31) at (-1, -8.7){
					$t_3, \LockColorOne{\lk_1}, \{ \LockColorTwo{\lk_2} \}$,
					$\sequence{e_{16}, e_{19}}$
				};
				\node[rectangle,align=center] (t32) at (2.5, -8.7){
					$t_3, \LockColorThree{\lk_3}, \{ \LockColorTwo{\lk_2} \}$,
					$\sequence{e_{13}}$
				};
				\draw[->, thick, -Latex] (t2) -- (t1);
				\draw[->, thick, -Latex] ([xshift=6mm]t1.south) to (t32.west);
				\draw[->, thick, -Latex, bend right=10] ([xshift=-3mm]t1.south) to ([xshift=-3mm]t31.north);
				\draw[->, thick, -Latex, bend right=10] (t31) to (t1);
\end{tikzpicture}
}
\end{subfigure}
%
\caption{
Abstract lock graphs of the traces from~\figref{motivating-no-dl} (left), \figref{motivating-dl} (middle) and \figref{syncp_example} (right).
\figlabel{motivating-graph}
}
 \end{figure}
\myparagraph{Abstract lock graph}{
The abstract lock graph of $\tr$ is a directed graph
$
\lkevgraph{\tr} = (V_\tr, E_\tr)
$,
where
\begin{itemize}
	\item $V_\tr=\{\tuple{t_1, \lk_1, L_1, F_1},\dots, \tuple{t_k, \lk_k, L_k, F_k}\}$ is the set of abstract acquires of $\tr$, and
	\item for every $\eta_1 {=} \tuple{t_1, \lk_1, L_1, F_1}, 
	\eta_2 {=} \tuple{t_2, \lk_2, L_2, F_2} \in V_\tr$, we have  $(\eta_1, \eta_2) \in E_\tr$ iff 
	$t_1 \neq t_2$, $\lk_1 \in L_2$, and $L_1 \cap L_2 = \emptyset$.
\end{itemize}
A node $\tuple{t_1, \lk_1, L_1, F_1}$ signifies that there is an event $\acq_1(\lk_1)$ performed by thread $t_1$ while holding the locks in $L_1$.
The last component $F_1$ is a list which contains all such events $\acq_1$ in order of appearance in $\tr$.
An edge $(\eta_1, \eta_2)$ signifies that the lock $\lk_1$ acquired by each of the
events $\acq_1 \in F_1$ was held by $t_2$ when it executed each of $\acq_2\in F_2$ while not holding a common lock.
The abstract lock graph can be constructed incrementally
as new events appear in $\tr$.
For $\NumEvents$ events, $\NumLocks$ locks and nesting depth $\NestingDepth$,
the graph has
$|V_\tr| = O\big(\NumThreads \cdot \NumLocks^{\NestingDepth}\big)$ vertices,
$|E_\tr| = O(|V_\tr|\cdot \NumLocks^{\NestingDepth - 1})$ edges and can be
constructed in $O(\NumEvents \cdot \NestingDepth)$ time.
See \figref{motivating-graph} for examples.
In the left graph, the cycle marks an abstract deadlock pattern and its single concrete deadlock pattern
$\abst{D}=\{ e_2 \} \times \{ e_{8} \}$, and similarly for the middle graph where $\abst{D}=\{ e_4 \} \times \{ e_{18} \}$.
In the right graph, there is a unique cycle which marks an abstract deadlock pattern of $6$ concrete deadlock patterns
$\abst{D}=\{ e_2, e_4, e_{29} \} \times \{ e_{16}, e_{19} \}$.

}
\setlength{\textfloatsep}{10pt}
\begin{algorithm}[h]
\small
\DontPrintSemicolon
\SetInd{0.4em}{0.4em}
\KwIn{
A trace $\tr$.
}
\KwOut{All abstract deadlock patterns of $\tr$ that contain a sync-preserving deadlock.}
\BlankLine
Construct the abstract lock graph $\lkevgraph{\tr}$\\
\ForEach{cycle $C=\tuple{\eta_0,\dots, \eta_{k-1}}$ in $G$}{
Let $\eta_i=\tuple{t_i, \lk_i, L_i, F_i}$\\
\uIf(\tcp*[h]{$C$ is an abstract deadlock pattern}){$\forall i\neq j$ we have $t_i\neq t_j$ and $\lk_i\neq \lk_j$ and $L_i\cap L_j=\emptyset$ }{
\lIf{$\checkAbsDeadP(C)$}{
Report that $C$ contains a sync-preserving deadlock
}
}
}
\caption{
Algorithm $\SyncPDOffline$.
\algolabel{offline}
}
\end{algorithm}
\normalsize

\myparagraph{Algorithm $\SyncPDOffline$}{
	It is straightforward to verify that every abstract deadlock pattern of $\tr$ appears as a (simple) cycle in $\lkevgraph{\tr}$.
	However, the opposite is not true.
	A cycle $C = \eta_0, \eta_1, \ldots, \eta_{k-1}$ of $\lkevgraph{\tr}$,
	where $\eta_i = \tuple{t_i, \lk_i, L_i, F_i}$ defines an abstract deadlock pattern
	if additionally every thread $t_i$ is distinct, all every lock $\lk_i$ is distinct, and all sets $L_i$ are pairwise disjoint.
	This gives us a simple recipe for enumerating all abstract deadlock patterns, by using
	Johnson's algorithm~\cite{Johnson1975} to enumerate
	every simple cycle $C$ in $\lkevgraph{\tr}$, and check whether $C$ is an abstract deadlock pattern.
	We thus arrived at our offline algorithm $\SyncPDOffline$ (\algoref{offline}).
	The running time depends linearly on the length of $\tr$
	and the number of cycles in $\lkevgraph{\tr}$.
}

\begin{restatable}{theorem}{syncpoffline}
\thmlabel{syncp_offline}
Consider a trace $\tr$ of $\NumEvents$ events, $\NumThreads$ threads and $\textsf{Cyc}_\tr$ cycles in $\lkevgraph{\tr}$.
The algorithm $\SyncPDOffline$ reports all sync-preserving deadlocks of $\tr$ in time $O(\NumEvents \cdot \NumThreads\cdot \textsf{Cyc}_\tr)$.
\end{restatable}

Although, in principle, we can have exponentially many cycles in $\lkevgraph{\tr}$,
because the nodes of $\lkevgraph{\tr}$ are \emph{abstract} acquire events (as opposed to \emph{concrete}), we expect that the number of cycles (and thus abstract deadlock patterns) in $\lkevgraph{\tr}$ remains small, even though the number of \emph{concrete} deadlock patterns can grow exponentially.
Since $\SyncPDOffline$ spends linear time per abstract deadlock pattern, we  have an efficient procedure overall 
for constant $\NumThreads$ and $\NumLocks$.
We evaluate $\textsf{Cyc}_\tr$ experimentally in \secref{experiments}, and confirm that it is very small compared to the number of concrete deadlock patterns in $\tr$.
Nevertheless, $\textsf{Cyc}_\tr$ can become exponential when $\NumThreads$ and $\NumLocks$ are large, making~\cref{algo:offline} run in exponential time.
Note that this barrier is unavoidable in general, as proven in \thmref{pattern-w1-hardness-pattern}.


\begin{example}
    \exlabel{ex4}
    We illustrate how the lock graph is integrated inside $\SyncPDOffline$.	
    Consider the trace $\tr_3$ in \figref{syncp_example}.
    It contains $6$ concrete deadlock patterns $D_1 \ldots D_6$.
    A naive algorithm would enumerate each pattern explicitly until it finds a deadlock. 
    However, the tight interplay between the abstract lock graph and sync-preservation enables a more efficient procedure.
    $\SyncPDOffline$ starts by computing the sync-preserving closure of $D_1$,
    $\SPClosure{\tr_3}(\prev{\tr_3}(\{ e_{2},e_{16}\}))=\set{e_1, \ldots, e_6, \, e_8, \ldots, e_{15}}$.
    As $e_2 \in \SPClosure{\tr_3}(\prev{\tr_3}(\{ e_{2},e_{16}\}))$, we conclude that $D_1$ is not a sync-preserving deadlock.
    The algorithm further deduces that the deadlock patterns $D_2$, $D_3$ and $D_4$ are also not sync-preserving deadlocks, as follows.
    $D_2=\pattern{e_2, e_{19}}$ shares a common event $e_2$ with $D_1$ but contains the event $e_{19}$ instead of $e_{16}$, while $e_5 \in \SPClosure{\tr_3}(\prev{\tr_3}(\{ e_{2},e_{16}\}))$.
    Since $e_{16} \tho{\tr_3} e_{19}$, and the sync-preserving closure grows monotonically (\propref{spclosure-monotone}), the sync-preserving closure of $e_2$ and $e_{19}$ will also contain $e_5$ (and thus $e_2$).
    Therefore, $D_2$ cannot be a sync-preserving deadlock.
    This reasoning is formalized in \corref{pattern-monotone}, and also applies to $D_3$ and $D_4$.
    Next, the algorithm proceeds with $D_5$.
    The above reasoning does not hold for $D_5$ as $\SPClosure{\tr_3}(\prev{\tr_3}(\{ e_{2},e_{16}\})) \cap S_5 = \emptyset$ 
    where $S_5=\set{e_{29}, e_{16}}$.
    The algorithm then computes the sync-preserving closure of $D_5$, reports a deadlock (\exref{syncp-ex}) and stops analyzing this abstract deadlock pattern.
    In the end, we have only explicitly enumerated the deadlock patterns $D_1$ and $D_5$.
\end{example}

\begin{remark}
Although the concept of lock graphs exists in the literature~\cite{Havelund2000,Bensalem2005,Cai2020,cai14magiclock},
our notion of \emph{abstract} lock graphs is novel and tailored to sync-preserving deadlocks.
The closest concept to abstract lock graphs is that of equivalent cycles~\cite{cai14magiclock}. 
However, equivalent cycles unify all the concrete patterns of a given abstract pattern and lead to unsound deadlock detection, which was indeed their use. 
\end{remark}


\section{On-the-fly Deadlock Prediction}
\seclabel{otf}

Although $\SyncPDOffline$ is efficient, both theoretically (\thmref{syncp_offline})
and in practice (Table~\ref{tab:expr-results}),
it runs in two passes, akin to other predictive deadlock-detection methods~\cite{Kalhauge2018,Cai2021}.
In a runtime monitoring setting, it is desirable to operate in an \emph{online} fashion.
Recall that $\checkAbsDeadP(\cdot)$ indeed operates online (\secref{verify-abstract-patterns}),
while the offline nature of $\SyncPDOffline$
is tied to the offline construction of the abstract lock graph $\lkevgraph{\tr}$.
To achieve the golden standard of online, linear-time, 
sound deadlock prediction, we focus on deadlocks of size $2$.
This focus is barely restrictive as most deadlocks in the wild have size $2$~\cite{Lu08}.
Further, deadlocks of size $2$ enjoy the following computational benefits:
\begin{enumerate*}[label=(\alph*)]
\item cycles of length $2$ can be detected instantaneously without performing graph traversals, and
\item \emph{every} cycle of length $2$ is an abstract deadlock pattern.
\end{enumerate*}

\myparagraph{Algorithm $\SyncPDOnline$}{
The algorithm $\SyncPDOnline$ maintains all abstract acquires of the form $\eta=\tuple{t, \lk_1, \{\lk_2\}, F}$,
i.e., we only focus on one lock $\lk_2$ that is protecting each such acquire.
When a new acquire event $e=\ev{t,\acq(\lk_1)}$ is encountered, 
the algorithm iterates over all the locks $\lk_2\in\lheld{\tr}(e)$ that are held in $e$,
and append the event $e$ to the sequence $F$ of the corresponding abstract acquire $\eta=\tuple{t, \lk_1, \{ \lk_2 \}, F}$; $F$ is maintained as FIFO queue.
Recall that we use timestamps on the acquire events in $F$ to determine membership in our closure computation.
Our online algorithm computes these timestamps on-the-fly and stores them in these queues together with the events.
Then the algorithm calls $\checkAbsDeadP(C)$ on the abstract deadlock pattern $C$ formed between $\eta$ and $\eta'=\tuple{t'\neq t, \lk_2, \{ \lk_1 \}, F'}$, in order to check for sync-preserving deadlocks between the deadlock patterns in $F\times F'$.
If $\checkAbsDeadP(C)$ reports no deadlock, the contents of $F'$ are emptied
as we are guaranteed that $F'$ will not cause a sync-preserving deadlock with any further acquire of thread $t$ on lock $\lk_1$.
For a trace with $\NumEvents$ events, $\NumThreads$ threads and $\NumLocks$ locks.
The algorithm calls $\checkAbsDeadP$ for each of the $O(\NumThreads^2\cdot \NumLocks^2)$ abstract
deadlock patterns of size $2$, each call taking $O(\NumEvents\cdot \NumThreads)$ time.


\small
\begin{algorithm*}[t!]
 \vspace*{-0.35cm}
\begin{multicols}{2}

\myfun{\checkDeadlock{$Lst$, $I$, $\tuple{t_1, \lk_1, t_2, \lk_2}$}}{
	\While{$\NOT\, \isEmpty{Lst}$}{
		($C_{\prev{}}, C$) := $\first{Lst}$ \;
		$I$ := \fixpoint{$I \mx C_{\prev{}}$, $\tuple{t_1, \lk_1, t_2, \lk_2}$} \;
		\lIf{$C \not\cle I$}{ \linelabel{check-containment}
			\declare `Deadlock' and \Break
		}
		$\removeFirst{Lst}$
	}
	\Return $I$
}

\myhandler{\wthandler{$t$, $x$}}{
	$\LW_x$ := $\Cc_t$ \linelabel{clkwt}\\
	$\Cc_t$ := $\Cc_t[t \mapsto \Cc_t(t)+1]$ \linelabel{clkctw}
}

\myhandler{\relhandler{$t$, $\lk$}}{
	$\Cc_t$ := $\Cc_t[t \mapsto \Cc_t(t)+1]$ \linelabel{clkrel} \;
	\ForEach{$t_1, t_2 \in \threads{}, \lk_1, \lk_2 \in  \locks{}$}{
		$\updateRel{\lst{\view{\AcqLst}{t, \lk}{\tuple{t_1, \lk_1, t_2, \lk_2}}}}{$\Cc_t$}$
	}
}

\myhandler{\acqhandler{$t$, $\lk$}}{
	$C_{\prev{}}$ := $\Cc_t$ \;
	$\Cc_t$ := $\Cc_t[t \mapsto \Cc_t(t)+1]$\linelabel{clkacq}\\ $\gId_\lk$ := $\gId_\lk + 1$ \linelabel{gidlk}\;
	\ForEach{$t_1, t_2 \in \threads{}, \lk_1, \lk_2 \in \locks{}$}{
		$\addLst{\view{\AcqLst}{t, \lk}{\tuple{t_1, \lk_1, t_2, \lk_2}}}{$(\gId_\lk, \Cc_t, \bot)$}$
	}
	\ForEach{$u\in \threads{}$, $\lk' \in \lheld{}$}{
		$\addLst{\view{\DPLst}{t, \lk, \lk'}{\tuple{u}}}{($C_{\prev{}}$, $\Cc_t$)}$ \linelabel{acqhist}
	}
	\ForEach{$u \neq t \in \threads{}$, $\lk' \in \lheld{}$}{
		$I$ := $\view{\Ii}{}{\tuple{u, \lk' t, \lk}} \mx C_{\prev{}}$ \;
		$\view{\Ii}{}{\tuple{u, \lk' t, \lk}}$ := \checkDeadlock{$\view{\DPLst}{u, \lk', \lk\,}{\tuple{t}}$, $I$, $\tuple{u, \lk' t, \lk}$} \linelabel{checkdlk}
	}

}

\myhandler{\rdhandler{$t$, $x$}}{
	$\Cc_t$ := $\Cc_t \mx \LW_x$ \linelabel{clkrd}
}

\BlankLine

\end{multicols}
 \vspace*{-0.25cm}
\normalsize
\caption{$\SyncPDOnline$.}
\algolabel{online}
\end{algorithm*}
\normalsize

$\SyncPDOnline$ is shown in detail in \algoref{online}.
The pseudocode contains handlers for processing the different events of $\tr$ in a streaming fashion, as well as a helper function for checking deadlocks.
The main data structures of the algorithm are
(i)~vector clocks $\Cc_t$, $\LW_x$, and $\view{\Ii}{}{\tuple{t_1, \lk_1, t_2, \lk_2}}$,
(ii)~scalar $\gId_\lk$, and
(iii)~FIFO queues of vector clocks $\view{\AcqLst}{t, \lk}{\tuple{t_1, \lk_1, t_2, \lk_2}}$ and $\view{\DPLst}{t, \lk_1, \lk_2}{\tuple{u}}$,
where $t, t_1, t_2, u$ range over threads, $x$ ranges over variables, and $\lk, \lk_1, \lk_2$ range over locks.
$\Cc_t$ stores the timestamp $\TS{\tr}^e$ where $e$ is the last event in thread $t$.
$\LW_x$ keeps track of the $\TS{\tr}^e$ where $e$ is the last event such that $\OpOf{e} = \wt(x)$. 
$\view{\Ii}{}{\tuple{t_1, \lk_1, t_2, \lk_2}}$ stores the computed sync-preserving closures for every tuple $\tuple{t_1, \lk_1, t_2, \lk_2}$.
The scalar variable $\gId_\lk$ keeps track of the index of the last acquire event on lock $\lk$. 
Similar to \cref{algo:compute-closure}, the FIFO queue $\view{\AcqLst}{t, \lk}{\tuple{t_1, \lk_1, t_2, \lk_2}}$ is maintained to keep track of the critical section history of thread $t$ and lock $\lk$.
Lastly, for an acquire event $e$, $\view{\DPLst}{t, \lk_1, \lk_2}{\tuple{u}}$ maintains a queue of tuples of the form $\tuple{C_{\prev{}}, \Cc_t}$ where $C_{\prev{}}$ and $\Cc_t$ are the timestamps of $\prev{\sigma}(e)$ and $e$, respectively
These tuples are utilized when checking for deadlocks (\cref{line:checkdlk}).
}

\begin{restatable}{theorem}{syncponline}
\thmlabel{syncp_online}
Consider a trace $\tr$ of $\NumEvents$ events, $\NumThreads$ threads and $\NumLocks$ locks.
The online $\SyncPDOnline$ algorithm reports all sync-preserving deadlocks of size $2$ of $\tr$ in $O(\NumEvents\cdot  \NumThreads^3\cdot \NumLocks^2)$ time.
\end{restatable}


\section{Experimental Evaluation}
\seclabel{experiments}

We first evaluated our algorithms
in an offline setting~(\secref{offline-expr}), where we record execution traces and evaluate different approaches on the \emph{same} input.
This eliminates biases due to non-deterministic thread scheduling.
Next, we consider an online setting~(\secref{online-expr}),
where we instrument programs and perform the analyses during runtime.
We conducted all our experiments on a standard laptop with \SI{1.8}{GHz} Intel Core i7 processor and \SI{16}{GB} RAM.

\subsection{Offline Experiments}
\seclabel{offline-expr}

\Paragraph{Experimental setup}
The goal of the first set of experiments is to evaluate 
$\SyncPDOffline$, and compare it
against prior algorithms for dynamic deadlock prediction.
In order for our evaluation to be precise we evaluate all algorithms on the \emph{same} execution trace.
We implemented $\SyncPDOffline$ in Java inside the \toolname analysis tool~\cite{rapid}, 
following closely the pseudocode in \algoref{offline}.
\toolname takes as input execution traces, as defined in \secref{prelim}.
These also include fork, join, and lock-request events.
We compare $\SyncPDOffline$ with two state-of-the-art, 
theoretically-sound albeit computationally more expensive, deadlock predictors,
\seqc~\cite{Cai2021} and \dirk~\cite{Kalhauge2018}, both of which also work on execution traces.

On the theoretical side, the complexity of \seqc is $\Otilde(\NumEvents^4)$, 
as opposed to the $\Otilde(\NumEvents)$ complexity of $\SyncPDOffline$. 
Moreover, \seqc only predicts deadlocks of size $2$, and though it could be extended to handle deadlocks of any size, this would degrade performance further.
\seqc may miss sync-preserving deadlocks even of size $2$, 
but can  detect deadlocks that are not sync-preserving.
Thus \seqc and $\SyncPDOffline$ are theoretically incomparable in their detection capability.
We refer to\begin{pldi}~\cite{arxiv}\end{pldi}\begin{arxiv}~\appref{incomp}\end{arxiv} for examples.
We noticed that \seqc fails on traces with non-well-nested locks --- we encountered one such case in our dataset.
\dirk's algorithm is theoretically complete, i.e., it can find all predictable deadlocks in a trace.
In addition, it can find deadlocks beyond the predictable ones, by reasoning about event values.
However, \dirk relies on heavyweight SMT-solving and
employs windowing techniques to scale to large traces. 
Due to windowing, it can miss deadlocks between events that are outside the given window. 
As with previous works~\cite{Cai2021, Kalhauge2018}, we set a window size of $10$K for \dirk.

Our dataset consists of several benchmarks 
from standard benchmark suites --- IBM Contest suite~\cite{Farchi03}, Java Grande suite~\cite{Smith01},
DaCapo~\cite{Blackburn06}, and
SIR~\cite{doESE05} ---
and recent literature~\cite{Kalhauge2018, Cai2021, jula2008deadlock, Joshi2009}.
Each benchmark was instrumented with RV-Predict~\cite{rvpredict} or Wiretap~\cite{Kalhauge2018} and
executed in order to log a single execution trace.


\begin{table}[h!]
\caption{
Trace characteristics, abstract lock graph statistics and performance comparison.
Columns 2-6 show the number of events, threads, variables, locks
and total number of lock acquire and request events.
Columns 7-9 show the number of cycles, abstract and concrete deadlock patterns in the abstract lock graph.
Columns 10 - 15 show the number of deadlocks reported and the times (in seconds) taken. 
by \dirk, \seqc, and \SyncPDOffline.
Time out (T.O) was set to $3$h.
F stands for technical failure.
\label{tab:expr-results}
}
\vspace{-0.17cm}
\setlength\tabcolsep{3pt}
\renewcommand{\arraystretch}{0.91}
\centering
\scalebox{0.86}{
\begin{tabular}{|r|c|c|c|c|c||c|c|c||c|c||c|c||c|c|}
\hline
1 & 2 & 3 & 4 & 5 & 6 & 7 & 8 & 9 & 10 & 11 & 12 & 13 & 14 & 15 \\
\hline
\multirow{2}{*}{\textbf{Benchmark}}& 
\multirow{2}{*}{$\mathcal{N}$} & \multirow{2}{*}{$\mathcal{T}$} & \multirow{2}{*}{$\mathcal{V}$} & \multirow{2}{*}{$\mathcal{L}$} & \multirow{2}{*}{$\NumAcquires/\mathcal{R}$} 
& \multicolumn{3}{c||}{ \textsf{A. Lock Graph }}
& \multicolumn{2}{c||}{ {\dirk} } 
& \multicolumn{2}{c||}{ {\seqc}} 
& \multicolumn{2}{c|}{ {\textsf{$\SyncPDOffline$}}} \\
\cline{7-15}
& & & & & 
& \textsf{|$\textsf{Cyc}$|}
& \textsf{\textsf{A. P.}}
& \textsf{\textsf{C. P.}}
& \textbf{Dlk} 
& {\textbf{Time}} 
&{\textbf{Dlk}} 
& {\textbf{Time}}
&{\textbf{Dlk}} 
&{\textbf{Time}} \\
\hline
Deadlock & 39 & 3 & 4 & 3 & 8 & 1 & 1 & 1 & 1 & 0.02 & 0 & 0.09 & 0 & 0.16\\
NotADeadlock & 60 & 3 & 4 & 5 & 16 & 1 & 1 & 1 & 0 & 0.02 & 0 & 0.09 & 0 & 0.16\\
Picklock & 66 & 3 & 6 & 6 & 20 & 2 & 2 & 2 & 1 & 0.02 & 1 & 0.10 & 1 & 0.18\\
Bensalem & 68 & 4 & 5 & 5 & 22 & 2 & 2 & 2 & 1 & 0.02 & 1 & 0.12 & 1 & 0.16\\
Transfer & 72 & 3 & 11 & 4 & 12 & 1 & 1 & 1 & 1 & 0.02 & 0 & 0.09 & 0 & 0.15\\
Test-Dimmunix & 73 & 3 & 9 & 7 & 26 & 2 & 2 & 2 & 2 & 0.02 & 2 & 0.10 & 2 & 0.17\\
StringBuffer & 74 & 3 & 14 & 4 & 16 & 1 & 3 & 6 & 2 & 0.02 & 2 & 0.12 & 2 & 0.19\\
Test-Calfuzzer & 168 & 5 & 16 & 6 & 48 & 2 & 1 & 1 & 1 & 0.02 & 1 & 0.12 & 1 & 0.17\\
DiningPhil & 277 & 6 & 21 & 6 & 100 & 1 & 1 & 3K & 1 & 1.60 & 0 & 0.09 & 1 & 0.17\\
HashTable & 318 & 3 & 5 & 3 & 174 & 1 & 2 & 43 & 2 & 0.19 & 2 & 0.12 & 2 & 0.19\\
Account & 706 & 6 & 47 & 7 & 134 & 3 & 1 & 12 & 0 & 0.19 & 0 & 0.09 & 0 & 0.18\\
Log4j2 & 1K & 4 & 334 & 11 & 43 & 1 & 1 & 1 & 1 & 0.65 & 1 & 0.11 & 1 & 0.20\\
Dbcp1 & 2K & 3 & 768 & 5 & 56 & 2 & 2 & 3 & - & F & 2 & 0.11 & 2 & 0.19\\
Dbcp2 & 2K & 3 & 592 & 10 & 76 & 1 & 2 & 4 & - & F & 0 & 0.10 & 0 & 0.18\\
Derby2 & 3K & 3 & 1K & 4 & 16 & 1 & 1 & 1 & 1 & 0.23 & 1 & 0.10 & 1 & 0.17\\
RayTracer & 31K & 5 & 5K & 15 & 976 & 0 & 0 & 0 & - & F & 0 & 0.15 & 0 & 0.19\\
jigsaw & 143K & 21 & 8K & 2K & 67K & 172 & 12 & 70 & - & F & 2 & 0.36 & 1 & 1.55\\
elevator & 246K & 5 & 727 & 52 & 48K & 0 & 0 & 0 & 0 & 1.65 & 0 & 0.33 & 0 & 0.27\\
hedc & 410K & 7 & 109K & 8 & 32 & 0 & 0 & 0 & 0 & 2.09 & 0 & 0.50 & 0 & 0.24\\
JDBCMySQL-1 & 442K & 3 & 73K & 11 & 13K & 2 & 4 & 6 & 2 & 28.45 & 2 & 0.24 & 2 & 0.48\\
JDBCMySQL-2 & 442K & 3 & 73K & 11 & 13K & 4 & 4 & 9 & 1 & 3.37 & 1 & 0.22 & 1 & 0.33\\
JDBCMySQL-3 & 443K & 3 & 73K & 13 & 13K & 5 & 8 & 16 & 1 & 31.23 & 1 & 0.25 & 1 & 0.45\\
JDBCMySQL-4 & 443K & 3 & 73K & 14 & 13K & 5 & 10 & 18 & 2 & 5.51 & 2 & 0.28 & 2 & 0.49\\
cache4j & 775K & 2 & 46K & 20 & 35K & 0 & 0 & 0 & 0 & 5.86 & 0 & 0.46 & 0 & 0.39\\
ArrayList & 3M & 801 & 121K & 802 & 176K & 9 & 3 & 672 & 3 & 8.7K & 3 & 21.98 & 3 & 1.68\\
IdentityHashMap & 3M & 801 & 496K & 802 & 162K & 1 & 3 & 4 & 1 & 443.93 & 1 & 8.51 & 1 & 1.45\\
Stack & 3M & 801 & 118K & 2K & 405K & 9 & 3 & 481 & 1 & T.O & 3 & 25.34 & 3 & 2.94\\
Sor & 3M & 301 & 2K & 3 & 719K & 0 & 0 & 0 & 0 & 15.89 & 0 & 44.12 & 0 & 0.61\\
LinkedList & 3M & 801 & 290K & 802 & 176K & 9 & 3 & 10K & 3 & 4.7K & 3 & 48.02 & 3 & 2.06\\
HashMap & 3M & 801 & 555K & 802 & 169K & 1 & 3 & 10K & 3 & 4.4K & 2 & 504.36 & 2 & 1.65\\
WeakHashMap & 3M & 801 & 540K & 802 & 169K & 1 & 3 & 10K & - & T.O & 2 & 499.68 & 2 & 1.70\\
Swing & 4M & 8 & 31K & 739 & 2M & 0 & 0 & 0 & - & F & 0 & 0.72 & 0 & 0.88\\
Vector & 4M & 3 & 15 & 4 & 800K & 1 & 1 & 1B & - & T.O & 1 & 1.52 & 1 & 1.90\\
LinkedHashMap & 4M & 801 & 617K & 802 & 169K & 1 & 3 & 10K & 2 & 40.74 & 2 & 492.87 & 2 & 1.69\\
montecarlo & 8M & 3 & 850K & 3 & 26 & 0 & 0 & 0 & 0 & 2.6K & 0 & 1.81 & 0 & 0.79\\
TreeMap & 9M & 801 & 493K & 802 & 169K & 1 & 3 & 10K & 2 & 105.45 & 2 & 480.11 & 2 & 1.92\\
hsqldb & 20M & 46 & 945K & 403 & 419K & 0 & 0 & 0 & - & F & - & - & 0 & 2.38\\
sunflow & 21M & 16 & 2M & 12 & 1K & 0 & 0 & 0 & - & F & 0 & 8.35 & 0 & 1.62\\
jspider & 22M & 11 & 5M & 15 & 10K & 0 & 0 & 0 & - & F & 0 & 8.49 & 0 & 1.95\\
tradesoap & 42M & 236 & 3M & 6K & 245K & 2 & 1 & 4 & - & F & 0 & 108.16 & 0 & 7.06\\
tradebeans & 42M & 236 & 3M & 6K & 245K & 2 & 1 & 4 & - & F & 0 & 116.23 & 0 & 7.26\\
eclipse & 64M & 15 & 10M & 5K & 377K & 9 & 5 & 280 & - & F & 0 & 26.67 & 0 & 9.90\\
TestPerf & 80M & 50 & 599 & 9 & 197K & 0 & 0 & 0 & 0 & 795.04 & 0 & 47.56 & 0 & 4.30\\
Groovy2 & 120M & 13 & 13M & 10K & 69K & 0 & 0 & 0 & 0 & 1.7K & 0 & 38.06 & 0 & 8.92\\
Tsp & 200M & 6 & 24K & 3 & 882 & 0 & 0 & 0 & 0 & 7.6K & 0 & 72.62 & 0 & 12.70\\
lusearch & 203M & 7 & 3M & 98 & 273K & 0 & 0 & 0 & 0 & 1.3K & 0 & 75.88 & 0 & 14.44\\
biojava & 221M & 6 & 121K & 79 & 16K & 0 & 0 & 0 & - & F & 0 & 63.79 & 0 & 12.65\\
graphchi & 241M & 20 & 25M & 61 & 1K & 0 & 0 & 0 & - & F & 0 & 102.05 & 0 & 25.25\\
\hline\hline\textbf{Totals} & \textbf{1B} & \textbf{7K} & \textbf{70M} & \textbf{37K} & \textbf{8M} & \textbf{256} & \textbf{93} & \textbf{1B} & \textbf{35} & \textbf{>18h} & \textbf{40} & \textbf{2801} & \textbf{40} & \textbf{135}\\\hline
\end{tabular}
}
\end{table}

\Paragraph{Evaluation}
\cref{tab:expr-results} presents our results.
A bug identifies a unique tuple of source
code locations corresponding to events participating in the deadlock.
Trace lengths vary vastly from $39$ to about $241$M, while the number of threads ranges from $3$ to about $800$,
which are representative features of real-world settings.
\texttt{Hsqldb} contains critical sections that are not well nested, 
and \seqc was not able to handle this benchmark;
our algorithm does not have such a restriction.

\vspace{-0.1cm}
\SubParagraph{\underline{Abstract vs Concrete Patterns}}
Columns 7-9 present statistics on the 
abstract lock graph $\lkevgraph{\tr}$ of each trace $\tr$.
Many traces have a large number of concrete deadlock patterns 
but much fewer abstract deadlock patterns;
a single abstract deadlock pattern can 
comprise up to an order of $10^4$ more concrete patterns (Column $8$ v/s Column $9$).
Unlike all prior sound techniques, 
our algorithms analyze 
abstract deadlock patterns, instead of concrete ones. 
We thus expect our algorithms to be much more scalable in practice.

\SubParagraph{\underline{Deadlock-detection capability}}
In total, both \seqc and $\SyncPDOffline$ reported 40 deadlocks.
\seqc misses a deadlock of size $5$ in \texttt{DiningPhil},
which is detected by $\SyncPDOffline$,
and $\SyncPDOffline$ misses a deadlock in \texttt{jigsaw} which is detected by \seqc.
As $\SyncPDOffline$ is complete for sync-preserving deadlocks, we conclude that there are no more such deadlocks in our dataset.
Overall, $\SyncPDOffline$ and \seqc miss only three deadlocks reported by \dirk. 
On closer inspection, we found that these deadlocks are not witnessed by correct reorderings, and require reasoning about event values.
On the other hand, \dirk struggles to analyze even moderately-sized benchmarks and times out in $3$ of them. 
This results in \dirk failing to report 5 deadlocks after $9$ hours, all of which are reported by $\SyncPDOffline$ in under a minute.
Similar conclusions were recently made in~\cite{Cai2021}.  
Overall, our results strongly indicate that the notion of sync-preservation characterizes most deadlocks that other tools are able to predict.

\SubParagraph{\underline{Unsoundness of \dirk}}
In our evaluation, we discovered that the soundness guarantee 
underlying \dirk~\cite{Kalhauge2018} is broken, resulting in it reporting false positives.
First, its constraint formulation~\cite{Kalhauge2018} 
does not rule out deadlock patterns when acquire events in the pattern hold common locks, 
in which case mutual exclusion disallows such a pattern to be a real predictable deadlock.
Second, \dirk also models conditional statements, allowing it to reason about witnesses beyond correct reorderings.
While this relaxation allows \dirk to predict additional deadlocks in \texttt{Transfer}, \texttt{Deadlock} and \texttt{HashMap}, 
its formalization is not precise and its implementation is erroneous.
We elaborate these aspects further in\begin{pldi}~\cite{arxiv}. \end{pldi}\begin{arxiv}~\appref{unsound-dirk}.\end{arxiv}


\SubParagraph{\underline{Running time}}
Our experimental results indicate that \dirk, backed by SMT solving, 
is the least efficient technique in terms of running time ---
it takes considerably longer or times out on large benchmark instances.
$\SyncPDOffline$ analyzed the entire set of traces $\sim\!\!\!21\times$ faster than \seqc.
On the most demanding benchmarks, such as 
HashMap and TreeMap, $\SyncPDOffline$ is more than $200\times$ faster than \seqc.
Although \seqc employs a polynomial-time algorithm for deadlock prediction, 
and thus significantly faster than the SMT-based \dirk,
the large polynomial complexity in its running time hinders scalability on 
execution traces coming from benchmarks that are more representative of realistic workloads.
In contrast, the linear time guarantees of $\SyncPDOffline$ are realized in practice, 
allowing it to scale on even the most challenging inputs.
More importantly, the improved performance comes while preserving essentially the same precision.



\SubParagraph{\underline{False negatives}}
Our benchmark set contains $93$ abstract deadlock patterns, $40$ of which are confirmed sync-preserving deadlocks.
We inspected the remaining $53$ abstract patterns to see if any of them are predictable deadlocks
missed by our sync-preserving criterion, independently of the compared tools.
$48$ of these $53$ patterns are in fact not predictable deadlocks ---
for every such pattern $D$, 
the set $S_D$ of events in the downward-closure of $\prev{}(D)$ with respect to $\tho{}$ and $\rf{}$,
already contains an event from $D$, disallowing any correct reordering
(sync-preserving or not) in which $D$ can be enabled.
Of the remaining, $4$ deadlock patterns obey the following scheme:
there are two acquire events $\acq_1, \acq_2$ participating in the deadlock pattern, 
each $\acq_i$ is preceded by a critical section on a lock that appears in 
$\lheld{}(\acq_{3-i})$, again disallowing a correct reordering that witnesses the pattern.
Thus, \emph{only one} predictable deadlock is not sync-preserving in our whole dataset.
This analysis supports that the notion of sync-preservation is not overly conservative in practice.

The above analysis concerns false negatives wrt. predictable deadlocks.
Some deadlocks are beyond the common notion of predictability we have adopted here, as they can only be exposed by reasoning about event values and control-flow dependencies, a problem that is $\NP$-hard even for 3 threads~\cite{Gibbons1997}.
We noticed $3$ such deadlocks in our dataset, found by \dirk,
though, as mentioned above, \dirk's reasoning for capturing such deadlocks is unsound in practice.

\subsection{Online Experiments}
\seclabel{online-expr}

\Paragraph{Experimental setup}
The objective of our second set of experiments is to evaluate 
the performance 
of our proposed algorithms in an \emph{online} setting.
For this, we implemented our $\SyncPDOnline$ algorithm inside the
framework of \dlfuzzer~\cite{Joshi2009} following closely the pseudocode in \algoref{online}. 
This framework instruments a concurrent program so that it can
perform analysis on-the-fly while executing it.
If a deadlock occurs during execution, it is reported and the execution halts.
However, if a deadlock is predicted in an alternate interleaving, 
then this deadlock is reported and the execution continues to search further deadlocks.
We used the same dataset as in \secref{offline-expr}, 
after discarding some benchmarks that could not be instrumented by \dlfuzzer.

To the best of our knowledge, all prior deadlock prediction techniques work offline.
For this reason, we only compared our online tool with the randomized 
scheduling technique of~\cite{Joshi2009} already
implemented inside the same \dlfuzzer framework.	
At a high level, this random scheduling technique works as follows.
Initially, it
(i)~executes the input program with a random scheduler, 
(ii)~constructs a \emph{lock dependency relation}, and 
(iii)~runs a cycle detection algorithm to discover deadlock patterns. 
For each deadlock pattern thus found, it spawns new executions that attempt 
to realize it as an actual deadlock.
To increase the likelihood of hitting the deadlock,
\dlfuzzer biases the random scheduler by pausing threads at specific locations.

The second, confirmation phase of~\cite{Joshi2009}
acts as a best-effort proxy for sound deadlock prediction.
On the other hand, $\SyncPDOnline$ is already sound and predictive, and thus does not require
additional confirmation runs, making it more efficient.
Towards effective prediction, we also implemented a simple bias to the scheduler.
If a thread $t$ attempts to write on a shared variable $x$ while holding a lock, then 
our procedure randomly decides to pause this operation for a short duration.
This effectively explores race conditions in different orders.
Overall, implementing $\SyncPDOnline$ inside \dlfuzzer provided the added advantage of supplementing a powerful prediction technique with a biased randomized scheduler.
To our knowledge, our work is the first to effectively 
combine these two orthogonal techniques.
We also remark that such a bias is of no benefit to \dlfuzzer itself
since it does not employ any predictive reasoning.


For this experiment, we run \dlfuzzer on each benchmark, and for each deadlock pattern found in the initial execution, 
we let it spawn $3$ new executions trying to realize the deadlock, 
as per standard (\href{https://github.com/ksen007/calfuzzer}{https://github.com/ksen007/calfuzzer}).
We repeated this process $50$ times and recorded the total time taken.
Then, we allocated the same time for $\SyncPDOnline$ to repeatedly execute the same program and perform deadlock prediction.
We measured all deadlocks found by each technique.


\begin{table*}
\caption{
Performance comparison of $\SyncPDOnline$ ($\syncpdshort$) and $\dlfuzzer$ ($\dlfshort$).
Columns 2-3 show the total number of bug reports. 
Columns 4-6 show the total number of unique bugs found by each tool, and their union. 
Columns 7-12 show the hit rate on each bug.
Columns 13-16 show the runtime overhead of the tools.
\label{tab:expr-dlf-results}
}
\vspace{-0.15cm}
\setlength\tabcolsep{3pt}
\renewcommand{\arraystretch}{0.9}
\small
\centering
\scalebox{1}{
\begin{tabular}{|r|c|c|c|c|c|c|c|c|c|c|c|c|c|c|c|c|c|c|c|c|c|c|c|c|c|c|c|}
\hline
1 & 2 & 3 & 4 & 5 & 6 & 7 & 8 & 9 & 10 & 11 & 12 & 13 & 14 & 15 & 16 \\
\hline
\multirow{2}{*}{\textbf{Benchmark}}& 
\multicolumn{2}{c|}{ {\textsf{Bug Hits}}} & 
\multicolumn{3}{c|}{ {\textsf{Unique Bugs}}} 
& \multicolumn{2}{c|}{ \textsf{Bug 1}} &
\multicolumn{2}{c|}{ \textsf{Bug 2}} &
\multicolumn{2}{c|}{ \textsf{Bug 3}} & \multicolumn{4}{c|}{ {\textsf{Runtime Overhead}}}  \\
\cline{7-16}
\cline{2-6}
& \textsf{$\syncpdshort$} & \textsf{$\dlfshort$} & \textsf{$\syncpdshort$} & \textsf{$\dlfshort$}  & All 
& \textsf{\textsf{$\syncpdshort$}} 
& \textbf{$\dlfshort$} 
& {\textbf{$\syncpdshort$}} 
&{\textbf{$\dlfshort$}} 
& {\textbf{$\syncpdshort$}}
&{\textbf{$\dlfshort$}} & {\textbf{$\syncpdshortinstr$}}  & {\textbf{$\syncpdshort$}}  & 
{\textbf{$\dlfshortinstr$}}  &
 {\textbf{$\dlfshort$}} \\
\hline
\hline
Deadlock & 50 & 50 & 1 & 1 & 1 & 50 & 50 & - & - & - & - & 2$\times$ & 3$\times$ & 2$\times$ & 4$\times$\\
Picklock & 227 & 97 & 2 & 1 & 2 & 226 & 97 & 1 & 0 & - & - & 2$\times$ & 2$\times$ & 2$\times$ & 5$\times$\\
Bensalem & 355 & 32 & 2 & 1 & 2 & 8 & 0 & 347 & 32 & - & - & 2$\times$ & 2$\times$ & 2$\times$ & 6$\times$\\
Transfer & 54 & 50 & 1 & 1 & 1 & 54 & 50 & - & - & - & - & 2$\times$ & 2$\times$ & 1$\times$ & 4$\times$\\
Test-Dimmunix & 702 & 0 & 2 & 0 & 2 & 351 & 0 & 351 & 0 & - & - & 2$\times$ & 2$\times$ & 2$\times$ & 4$\times$\\
StringBuffer & 153 & 131 & 2 & 2 & 2 & 128 & 118 & 25 & 13 & - & - & - & - & - & -\\
Test-Calfuzzer & 177 & 44 & 1 & 1 & 1 & 177 & 44 & - & - & - & - & 2$\times$ & 2$\times$ & 2$\times$ & 4$\times$\\
DiningPhil & 162 & 100 & 1 & 1 & 1 & 162 & 100 & - & - & - & - & - & - & - & -\\
HashTable & 169 & 120 & 2 & 2 & 2 & 82 & 21 & 87 & 99 & - & - & - & - & - & -\\
Account & 19 & 188 & 1 & 1 & 1 & 19 & 188 & - & - & - & - & 2$\times$ & 8$\times$ & 2$\times$ & 16$\times$\\
Log4j2 & 290 & 100 & 2 & 1 & 2 & 145 & 100 & 145 & 0 & - & - & - & - & - & -\\
Dbcp1 & 265 & 138 & 2 & 2 & 2 & 264 & 61 & 1 & 77 & - & - & - & - & - & -\\
Dbcp2 & 129 & 126 & 2 & 2 & 2 & 86 & 99 & 43 & 27 & - & - & - & - & - & -\\
RayTracer & 0 & 0 & 0 & 0 & 0 & - & - & - & - & - & - & 122$\times$ & 124$\times$ & 109$\times$ & 111$\times$\\
Tsp & 0 & 0 & 0 & 0 & 0 & - & - & - & - & - & - & 47$\times$ & 60$\times$ & 37$\times$ & 40$\times$\\
jigsaw & 1189 & 1 & 1 & 1 & 2 & 1189 & 0 & 0 & 1 & - & - & - & - & - & -\\
elevator & 0 & 0 & 0 & 0 & 0 & - & - & - & - & - & - & 2$\times$ & 2$\times$ & 2$\times$ & 2$\times$\\
JDBCMySQL-1 & 349 & 117 & 2 & 3 & 3 & 1 & 21 & 0 & 4 & 348 & 92 & 3$\times$ & 4$\times$ & 2$\times$ & 13$\times$\\
JDBCMySQL-2 & 559 & 73 & 1 & 1 & 1 & 559 & 73 & - & - & - & - & 2$\times$ & 4$\times$ & 2$\times$ & 18$\times$\\
JDBCMySQL-3 & 560 & 224 & 1 & 1 & 1 & 560 & 224 & - & - & - & - & 2$\times$ & 5$\times$ & 2$\times$ & 24$\times$\\
JDBCMySQL-4 & 1717 & 101 & 3 & 1 & 3 & 95 & 0 & 834 & 0 & 788 & 101 & 3$\times$ & 5$\times$ & 2$\times$ & 31$\times$\\
hedc & 0 & 0 & 0 & 0 & 0 & - & - & - & - & - & - & 2$\times$ & 2$\times$ & 1$\times$ & 2$\times$\\
cache4j & 0 & 0 & 0 & 0 & 0 & - & - & - & - & - & - & 2$\times$ & 2$\times$ & 2$\times$ & 2$\times$\\
lusearch & 0 & 0 & 0 & 0 & 0 & - & - & - & - & - & - & 16$\times$ & 17$\times$ & 13$\times$ & 16$\times$\\
ArrayList & 47 & 45 & 3 & 3 & 3 & 20 & 22 & 3 & 5 & 24 & 18 & 50$\times$ & 69$\times$ & 18$\times$ & 79$\times$\\
Stack & 44 & 27 & 3 & 3 & 3 & 18 & 13 & 8 & 4 & 18 & 10 & 69$\times$ & 91$\times$ & 64$\times$ & 86$\times$\\
IdentityHashMap & 68 & 62 & 2 & 2 & 2 & 13 & 47 & 55 & 15 & - & - & 4$\times$ & 8$\times$ & 3$\times$ & 10$\times$\\
LinkedList & 48 & 26 & 3 & 2 & 3 & 21 & 17 & 7 & 0 & 20 & 9 & 16$\times$ & 28$\times$ & 14$\times$ & 32$\times$\\
Swing & 0 & 0 & 0 & 0 & 0 & - & - & - & - & - & - & 5$\times$ & 6$\times$ & 4$\times$ & 6$\times$\\
Sor & 0 & 0 & 0 & 0 & 0 & - & - & - & - & - & - & 2$\times$ & 7$\times$ & 2$\times$ & 2$\times$\\
HashMap & 46 & 44 & 2 & 2 & 2 & 18 & 11 & 28 & 33 & - & - & 7$\times$ & 11$\times$ & 4$\times$ & 8$\times$\\
Vector & 126 & 50 & 1 & 1 & 1 & 126 & 50 & - & - & - & - & 2$\times$ & 2$\times$ & 2$\times$ & 3$\times$\\
LinkedHashMap & 57 & 43 & 2 & 2 & 2 & 22 & 10 & 35 & 33 & - & - & 10$\times$ & 10$\times$ & 4$\times$ & 8$\times$\\
WeakHashMap & 29 & 40 & 2 & 2 & 2 & 6 & 11 & 23 & 29 & - & - & 7$\times$ & 12$\times$ & 4$\times$ & 8$\times$\\
montecarlo & 0 & 0 & 0 & 0 & 0 & - & - & - & - & - & - & 16$\times$ & 100$\times$ & 13$\times$ & 126$\times$\\
TreeMap & 42 & 47 & 2 & 2 & 2 & 16 & 15 & 26 & 32 & - & - & 9$\times$ & 12$\times$ & 5$\times$ & 9$\times$\\
eclipse & 0 & 0 & 0 & 0 & 0 & - & - & - & - & - & - & 2$\times$ & 2$\times$ & 2$\times$ & 2$\times$\\
TestPerf & 0 & 0 & 0 & 0 & 0 & - & - & - & - & - & - & 2$\times$ & 2$\times$ & 2$\times$ & 2$\times$\\
\hline\hline\textbf{Total} & \textbf{7633} & \textbf{2076} & \textbf{49} & \textbf{42} & \textbf{51} & - & - & - & - & - & - & - & - & - & - \\ 
\hline
\end{tabular}
}
\end{table*}

\Paragraph{Evaluation}
\cref{tab:expr-dlf-results} presents our experimental results.
Columns $2$-$3$ of the table display the total number of bug hits,
which is the total number of times a bug was predicted by $\SyncPDOnline$ in the entire duration,
or was confirmed in any trial of \dlfuzzer.
Columns $4$-$6$ display the unique bugs (i.e., unique tuples of source code locations leading to a deadlock) 
found by the techniques.
The employed techniques are able to find a maximum of $3$ unique bugs for each benchmark
in our benchmark set. 
Respectively, columns $7$-$12$ display the detailed information on the number 
of times a particular bug was found by each technique.
Runtime overheads are displayed in the columns $13$-$16$, with $\mathsf{\tt -I}$ denoting the instrumentation phase only.

\SubParagraph{\underline{Deadlock-detection capability}}
\dlfuzzer had $2076$ bug reports in total, and it found $42$ unique bugs.
In contrast, $\SyncPDOnline$ flagged $7633$ bug reports, corresponding to $49$ unique bugs.
In more detail, \dlfuzzer missed $9$ bugs reported by \SyncPDOnline whereas 
$\SyncPDOnline$ missed $2$ bugs reported by \dlfuzzer.
Also, \SyncPDOnline significantly outperformed \dlfuzzer in total number of bugs hits.
Our experiments again support that the notion of sync-preservation  captures most deadlocks that occur in practice, to the extent that other state-of-the-art techniques can capture.
A further observation is that in the offline experiments, \SyncPDOffline  was not able to find deadlocks in \texttt{Transfer} and \texttt{Deadlock}. 
However, the random scheduling procedure allowed \SyncPDOnline 
to navigate to executions from which deadlocks can be predicted.
This demonstrates the potential of combining predictive dynamic
techniques with controlled concurrency testing.

\SubParagraph{\underline{Runtime overhead}}
We have also measured the runtime overhead of both \SyncPDOnline and \dlfuzzer,
both as incurred by instrumentation, as well as by the deadlock analysis.
The latter is the time taken by \algoref{online} for the case of $\SyncPDOnline$,
and the overhead introduced due to the new executions in the second confirmation phase for the case of \dlfuzzer.
Our results show that the instrumentation overhead of \SyncPDOnline is, in fact, comparable to that of \dlfuzzer, though somewhat larger. 
This is expected, as \SyncPDOnline needs to also instrument memory access events, while \dlfuzzer only instruments lock events, but at the same time surprising because the number of memory access events
is typically much larger than the number of lock events.
On the other hand, the analysis overhead is often larger for \dlfuzzer, 
even though it reports fewer bugs.
It was not possible to measure the runtime overhead in certain benchmarks as 
either they were always deadlocking or the computation was running indefinitely.


\section{Related Work}
\seclabel{related_work}





Dynamic techniques for detecting deadlock patterns,
like the GoodLock algorithm~\cite{Havelund2000}
have been improved in performance~\cite{Cai2020,zhou2017undead}
and precision~\cite{Bensalem2005}, sometimes using re-executions to verify potential deadlocks~\cite{Samak2014,Samak2014b,Sorrentino2015,Joshi2009,Bensalem2006}.
%
Predictive analyses directly infer concurrency bugs in
alternate executions~\cite{serbanuta2013} and are typically
\emph{sound} (no false positives).
This approach has been successfully applied for
detecting bugs 
such as data races~\cite{Said11,Huang14,Smaragdakis12,Kini2017,Pavlogiannis2020,Mathur2021,Roemer20}, use-after-free vulnerabilities~\cite{Huang2018}, and more recently for deadlocks~\cite{Eslamimehr2014,Kalhauge2018,Cai2021}.



The notion of sync-preserving deadlocks has been inspired by 
a similar notion pertaining to data races~\cite{Mathur2021}.
However, sync-preserving deadlock prediction rests on some further novelties.
First, unlike data races, deadlocks can involve more than $2$ events.
Generalizing sync-preserving ideals of sets of events of arbitrary size, as well as 
establishing the monotonicity properties (\propref{spclosure-monotone} and \corref{pattern-monotone}) for arbitrarily many events is non-trivial.
Second, our notions of abstract deadlock patterns (\secref{verify-abstract-patterns}) 
and abstract lock graphs (\secref{enumerate-patterns}) are novel and carefully crafted
to leverage these monotonicity properties in the deadlock setting.
Indeed, the linear-time sync-preserving verification of each abstract deadlock pattern is the cornerstone of our approach, for the first linear-time, sound and precise deadlock predictor.

Although the basic principles of data-race and deadlock prediction are similar, there are notable differences.
First, identifying potential deadlocks is theoretically intractable, whereas, potential races are identified easily. 
Second, popular partial-order based techniques~\cite{Flanagan09,Kini2017} for data races are likely to require non-trivial modifications for deadlocks, as they typically order critical sections, which may hide a deadlock.
Nevertheless, bridging prediction techniques between data races and deadlocks is an interesting and relatively open direction.

Predicting deadlocks is an intractable problem, the complexity of which we have characterized in this work. 
Prior works have also focused on the complexity of predicting data races~\cite{Mathur2020b,Kulkarni2021} and atomicity violations~\cite{Farzan2099}.

\section{Conclusion}
\seclabel{conclusions}
We have studied the complexity of deadlock prediction and
introduced the new tractable notion of sync-preserving deadlocks,
along with sound, complete and efficient algorithms
for detecting them. 
Our experiments show that the majority of deadlocks occurring in practice are indeed sync-preserving, and our algorithm $\SyncPDOffline$ is the first deadlock predictor that achieves 
sound and high coverage, while also spending only linear time to process its input.
Our online algorithm $\SyncPDOnline$ enhances the bug detection
capability of controlled concurrency testing techniques like~\cite{Joshi2009},
at close runtime overheads.
Interesting future work includes incorporating static checks~\cite{bigfoot2017}
and
exploring ways for deeper integration of controlled concurrency testing
with predictive techniques.
Another step is to extend the coverage of sync-preserving deadlocks while maintaining efficiency, for example, by reasoning about program control flow.



\begin{acks}
Andreas Pavlogiannis was partially supported by a research grant (VIL42117) from VILLUM FONDEN.
Umang Mathur was partially supported by the Simons Institute for the Theory of Computing, and by a Singapore Ministry of Education (MoE) Academic Research Fund (AcRF) Tier 1 grant.
Mahesh Viswanathan was partially supported by NSF SHF 1901069 and NSF CCF 2007428.
\end{acks}

\section*{Data and Software Availability Statement}

The artifact developed for this work is available~\cite{artifact}, which contains all source codes and experimental data necessary to
reproduce our evaluation in~\cref{sec:experiments}, excluding the results of \seqc.


\bibliography{references}



\begin{arxiv}
\pagebreak
\appendix

\section{Proofs of \secref{lower-bounds}}\seclabel{sec:app_proofs_lower_bounds}

In this section we present rigorous proofs of our hardness results in \thmref{pattern-w1-hardness-pattern}, \thmref{pattern-ov-hardness}, and \thmref{w1-hardness-pattern}.


\subsection{Proof of \thmref{pattern-w1-hardness-pattern}}
\seclabel{w1-hardness-pattern}

Checking whether a trace of length
$\NumEvents$ has a deadlock pattern of size $k$
can be solved using a simple enumeration based algorithm
in time $O(\poly{\NumThreads} \cdot \NumEvents^\NumThreads)$.
In this section, we show that this problem is, in fact,
$\W{1}$-hard
and thus is unlikely to be FPT (fixed parameter tractable).
In other words, an algorithm with running time $O(f(\NumThreads) \cdot \poly{\NumEvents})$ 
is unlikely to exist for this problem, for any computable function $f$.

\patternwonehardness*
\begin{proof}
	We show that there is a polynomial-time FPT-reduction from
	INDEPENDENT-SET(c) to the problem of checking the existence of deadlock-patterns
	of size $c$.
	Our reduction takes as input an undirected graph $G$ and outputs a trace $\tr$
	such that $G$ has an independent set of size $c$ iff $\tr$
	has a deadlock pattern of size $c$.
	

	\myparagraph{Construction}{
		We assume  that the vertices are indexed from $1$ through $n = |V|$: $V = \set{v_1, v_2, \ldots, v_n}$.
		We also assume a total ordering $<_E$ on the set of edges $E$.
		The trace $\tr$ we construct is a concatenation of $c$ sub-traces: 
		$
		\tr = \tr^{(1)} \cdot \tr^{(2)} \cdots \tr^{(c)}
		$
		and uses $c$ threads $\set{t_1, t_2, \ldots t_c}$ and
		$|E| + c$ locks $\set{\lk_{\set{u, v}}}_{\set{u, v} \in E} \uplus \set{\lk_0, \lk_1 \ldots, \lk_{c-1}}$.
		The $i^\text{th}$ sub-trace $\tr^{(i)}$ is a sequence of events performed by thread $t_i$, and
		is obtained by concatenation of $n = |V|$ sub-traces:
		$
		\tr^{(i)} = \tr^{(i)}_1 \cdot \tr^{(i)}_2 \cdots \tr^{(i)}_n
		$.
		Each sub-trace $\tr^{(i)}_j$ with $(i \leq c, j \leq n)$ 
		comprises of nested critical sections over locks of the 
		form $\lk_{\set{v_j, u}}$, where $u$ is a neighbor of $v_j$.
		Inside the nested block we have critical 
		sections on locks $\lk_{i \% c}$ and $\lk_{(i+1) \% c}$.
		Formally, let $\set{v_j, v_{k_1}}, \ldots, \set{v_j, v_{k_d}}$
		be the neighboring edges of $v_j$ (ordered according to $<_E$).
		Then, $\tr^{(i)}_j$ is the unique string generated by the
		grammar having $d+1$ non-terminals $S_0, S_1, \ldots, S_d$, start symbol $S_d$
		and the following production rules:
		\begin{itemize}
			\item $S_0 \to \ev{t_i, \acq(\lk_{i \% c})} \cdot \ev{t_i, \acq(\lk_{(i+1) \% c})} \cdot \ev{t_i, \rel(\lk_{(i+1) \% c})} \cdot \ev{t_i, \rel(\lk_{i \% c})}$.
			\item for each $1 \leq r \leq d$, $S_r \to \ev{t_i, \acq(\lk_{\set{v_j, v_{k_r}}})} \cdot S_{r-1} \cdot \ev{t_i, \rel(\lk_{\set{v_j, v_{k_r}}})}$.
		\end{itemize}
		\figref{w1-hardness} illustrates this construction for a graph with $3$ nodes
		and parameter $c = 3$.
		Finally, observe that the lock-nesting depth in $\tr$ is bounded by 2 + the degree of $G$.
	}
	
\myparagraph{Correctness}{
	We now argue for the correctness of the construction.
	First, consider the case when $G$ has an independent set $I = \set{v_{j_1}, \ldots v_{j_c}}$ of size $\geq c$.
	Let  be $c$ distinct vertices from the independent set.
	Let $e^{(i)}_{j}$ be the innermost acquire (on lock $\lk_{(i+1) \% c}$) in the sub-trace 
	$\tr^{(i)}_{j}$.
	We show that the sequence $D = \pattern{e^{(1)}_{j_1}, e^{(2)}_{j_2} \ldots, e^{(c)}_{j_c}}$ 
	is a deadlock pattern.
	Observe that $\ThreadOf{e^{(i)}_{j_i}} = t_i \neq t_{i'} = \ThreadOf{e^{(i')}_{j_{i'}}}$ for every $i \neq i'$.
	Similarly, the locks acquired in $e^{(i)}_{j_i}$ and $e^{(i')}_{j_{i'}}$ are also distinct
	when $i \neq i'$.
	Finally, we want to show that $\lheld{\tr}(e^{(i)}_{j_i}) \cap \lheld{\tr}(e^{(i')}_{j_{i'}}) = \emptyset$ for every $i \neq i'$.
	Assume on the contrary that there is a lock $\lk \in \lheld{\tr}(e^{(i)}_{j_i}) \cap \lheld{\tr}(e^{(i')}_{j_{i'}})$.
	Clearly, $\lk$ cannot be of the form $\lk_m$ for some $0 \leq m < c$ as the only such lock held
	at $e^{(i)}_{j_i}$ is $\lk_{i \% c}$ and the only such lock held at $e^{(i')}_{j_{i'}}$ is $\lk_{i' \% c}$
	which are different.
	Thus, it must be of the form $\lk_{\set{u, v}}$ for some $\set{u, v} \in E$.
	Since it is acquired in both sub-traces  $\tr^{(i)}_{j_i}$ and $\tr^{(i')}_{j_{i'}}$,
	we have $\lk = \lk_{\set{v_i, v_{i'}}}$.
	This means that $\set{v_i, v_{i'}} \in E$ contradicting that $I$ is an independent set.

	Now, consider the case when there is a deadlock pattern of size $c$:
	$D = \pattern{e_0, e_1, \ldots, e_{c-1}}$.
	Since there are exactly $c$ threads in $\tr$, there must be one event in each thread.
	We first argue that if there is one $e_i$ that acquires a lock of the form $\lk_{\set{u, v}}$ then all events in the deadlock pattern do so.
	This follows from a simple inductive argument --- $e_{i+1}$ must acquire a lock held at $e_i$
	and thus must be of the form $\lk_{\set{u', v'}}$, and so on.
	Next, all locks of the form $\lk_{\set{u, v}}$ are acquired in the order consistent with
	the total order $<_E$.
	Hence, they cannot form a cyclic dependency.
	Thus, the locks acquired at $e_0, \ldots e_{c-1}$ are $\lk_0, \ldots, \lk_{c-1}$.
	Symmetry ensures that, you can assume, w.l.o.g that $\OpOf{e_i} = \acq(\lk_{i \% c})$.
	Let $f(i)$ be the index such that $e_i$ is in the subtrace $e^{(i)}_{f(i)}$.
	Then, we show that the set $A = \set{v_{f(1)}, v_{f(2)}, \ldots, v_{f(c)}}$ is an independent set of size $c$.
	Assume on the contrary that there are $i \neq i'$ such that $\set{v_{f(i)}, v_{f(i')}} \in E$.
	Then, our construction ensures that
	the lock $\lk_\set{v_{f(i)}, v_{f(i')}}$ will be held at both $e_i$ and $e_{i'}$
	contradicting that $D$ is a deadlock pattern.
}	

\myparagraph{Time complexity}{
		The size of the trace is $O\big(c \cdot (2 + \sum_{v \in V} \mathsf{degree}(v)) \big)$
		$ = O(c \cdot (|V| + |E|))$
		and the time taken is also the same.
	}
\end{proof}


\subsection{Proof of \thmref{pattern-ov-hardness}}
\seclabel{ov-hardness}

\patternovhardness*
\begin{proof}
	We show a fine-grained reduction from the Orthogonal Vectors Problem to
	the problem of checking for deadlock patterns of size $2$.
	For this, we start with two sets
	$A, B \subseteq \set{0, 1}^d$
	of $d$-dimensional vectors with $|A| = |B| = n$.
	We write the $i^{th}$ vector in $A$ as $A_i$ and that in $B$ as $B_i$.

	\myparagraph{Construction}{
		We will construct a trace $\tr$ such that $\tr$ has a deadlock
		pattern of length $2$ iff $(A, B)$ is a positive OV instance.
		The trace $\tr$ is of the form $\tr = \tr^A \cdot \tr^B$
		and uses $2$ threads $\set{t_A, t_B}$ and $d+2$ distinct locks $\lk_1, \ldots, \lk_d, \lk, \lk'$.
		$\tr^A$ and $\tr^B$, intuitively encode the given sets of vectors $A$ and $B$.
		The sub-traces $\tr^A = \tr^A_1 \cdot \tr^A_2 \cdots \tr^A_n$
		and $\tr^B = \tr^B_1 \cdot \tr^B_2 \cdots \tr^B_n$ are defined as follows.
		For each $i \in \set{1, 2, \ldots, n}$ and $Z \in \set{A, B}$, the
		sub-trace $\tr^Z_i$ is the unique string generated by the
		grammar having $d+1$ non-terminals $S_0, S_1, \ldots, S_d$, start symbol $S_d$
		and the following production rules:
		\begin{itemize}
			\item $S_0 \to \ev{t_Z, \acq(m)} \cdot \ev{t_Z, \acq(m')} \cdot \ev{t_Z, \rel(m')} \cdot \ev{t_Z, \rel(m)}$,
			where $(m, m') = (\lk, \lk')$ if $Z = A$, and $(m, m') = (\lk', \lk)$ otherwise.
			\item for each $1 \leq j \leq d$, $S_j \to S_{j-1}$ if $Z_i[j] = 0$.
			Otherwise (if $Z_i[j] =1$), $S_j \to \ev{t_Z, \acq(\lk_j)} \cdot S_{j-1} \cdot \ev{t_Z, \rel(\lk_j)}$.
		\end{itemize}
		In words, all events of $\tr^A$ are performed by thread $t_A$ and those
		in $\tr^B$ are performed by $t_B$.
		Next, the $i^{th}$ sub-trace of $\tr^A$, denoted $\tr^A_i$ corresponds to the vector $A_i$
		as follows --- $\tr^A_i$ is a nested block of critical sections,
		with the inner most critical section being on lock $\lk'$,
		which is immediately enclosed in a  critical section on lock $\lk$.
		Further, in the sub-trace $\tr^A_i$, the lock $\lk_j$
		occurs iff $A_i[j] = 1$.
		The sub-traces $\tr^B_i$ is similarly constructed, except that the order
		of the two inner most critical sections is inverted.
		\figref{ov-hardness} illustrates the construction for an OV-instance with $n=2$ and $d=2$.
	}

\myparagraph{Correctness}{
We now argue for the correctness of the construction.
First, if $(A, B)$ is a positive $OV$ instance, then
there are two indices $\alpha, \beta \in \set{1 ,\ldots, n}$
such that $A_\alpha$ and $B_\beta$ are orthogonal.
Thus, for every $1 \leq j \leq d$, we have $A_\alpha[j] = 0$
or $B_\beta[j] = 0$ (possibly both).
Based on the construction, this implies that for every lock $\lk_j$
(with $1 \leq j \leq d$),
$\lk_j$ occurs in at most one of the sub-traces $\tr^A_\alpha$ and $\tr^B_\beta$.
This gives a deadlock pattern $D = \pattern{e, e'}$,
where $e$ is the unique event in the sub-trace $\tr^A_\alpha$ with $\OpOf{e} = \acq(\lk')$
and $e'$ is the unique event in $\tr^B_\beta$ with $\OpOf{e'} = \acq(\lk)$.
This is because $\lk \in \lheld{\tr}(e)$, $\lk' \in \lheld{\tr}(e')$,
and each lock $\lk_j$ exists in at most one of $\lheld{\tr}(e)$
or $\lheld{\tr}(e')$ (and thus $\lheld{\tr}(e) \cap \lheld{\tr}(e') = \emptyset$).

Now, consider the case when there is a deadlock pattern $D = \pattern{e, e'}$ in $\tr$.
W.l.o.g, we can assume that
$\ThreadOf{e} = t_A$ and $\ThreadOf{e'} = t_B$.
This means there is a $\alpha$ and $\beta$ such that
$e$ occurs in $\tr^A_\alpha$ and $e'$ occurs in $\tr^B_\beta$.
Next, observe that $e$ and $e'$ can only be acquire events on locks $\set{\lk, \lk'}$
and not on the locks $\set{\lk_1, \ldots, \lk_d}$.
This is because the order of acquisition amongst $\lk_1, \ldots, \lk_d$
is always fixed and further, each of these are never acquired inside of critical sections
of $\lk$ or $\lk'$.
By the choice of threads, $\OpOf{e} = \acq(\lk')$
and $\OpOf{e'} = \acq(\lk)$.
Since $\lheld{\tr}(e) \cap \lheld{\tr}(e') = \emptyset$, 
for every $j$, $\lk_i \not\in \lheld{\tr}(e) \cap \lheld{\tr}(e')$.
This means that at most one of $A_\alpha[j] = 1$ an $B_\beta[j] = 1$ is true.
This means that the vectors $A_\alpha$ and $B_\beta$ are orthogonal.
}

	\myparagraph{Time complexity}{
		The total size of the trace is $\NumEvents = O(n\cdot d)$
		and it has $\NumLocks = d + 2 = O(d)$ locks.
		The construction can be performed in time $O(n\cdot d)$.
		Thus, if for some $\epsilon > 0$, 
		there is a $O(\poly{\NumLocks} \cdot \NumEvents^{2-\epsilon}) = O(\poly{d} \cdot n^{2-\epsilon})$ 
		algorithm for the problem of detecting deadlock patterns of length $2$,
		we would get a $O(\poly{d} \cdot n^{2-\epsilon} + n\cdot d) = O(\poly{d} \cdot n^{2-\epsilon})$
		time algorithm for OV, falsifying the OV hypothesis.
	}
\end{proof}


\subsection{$\W{1}$-hardness for Deadlock Prediction}
\seclabel{w1-hardness}

Finally, we prove that predicting deadlocks is a $\W{1}$-hard problem parameterized by the number of threads $\NumThreads$.

\wonehardness*
\begin{proof}
    We establish a simple reduction from the problem of data-race prediction to the problem of predicting deadlocks of size $2$.
	Given a trace $\tr$ and a deadlock pattern $D=\pattern{\acq_1, \acq_2}$ of size $2$, we choose a fresh variable $x\not \in \vars{\tr}$, and construct two events $\wt_1=\ev{\ThreadOf{\acq_1}, \wt(x)}$ and $\wt_2=\ev{\ThreadOf{\acq_2}, \wt(x)}$, where $x$.
	We construct a trace $\tr'$ by replacing $\acq_1$ and $\acq_2$ in $\tr$ with $\wt_1$ and $\wt_2$, respectively.
	It is straightforward to verify that any correct $\rho$ reordering of $\tr'$ that witnesses $(\wt_1, \wt_2)$ as a data race is also a correct reordering of $\tr$ that witnesses the deadlock $(\acq_1, \acq_2)$.
	Since predicting data races is $\W{1}$-hard on $\NumThreads$\cite{Mathur2020b}, the desired result follows.
\end{proof}


\section{Proofs of \secref{syncp}}
\seclabel{sec:app_proofs_characterize_patterns}

We first prove \lemref{spclosure-necessary-sufficient}.

\spclosureNecessarySufficient*


\begin{proof}
Let $\rho$ be a sync-preserving correct reordering of $\tr$ (and thus $\events{\rho} \subseteq \events{\tr}$)
such that $S \subseteq  \events{\rho}$.
Since, $\rho$ is a correct-reordering of $\tr$, for every event $e \in \events{\rho}$,
we must have $f \in \events{\rho}$ for every $f \tho{\tr} e$.
Likewise, for every read event $e \in \events{\rho}$, we must also have $\rf{\tr}(e) \in \events{\rho}$,
again because $\rho$ is a correct reordering of $\tr$.
Finally, let $e \in \events{\rho}$ be an acquire event (with $\OpOf{e} = \acq(\lk)$), and let
$f \neq e$ be another acquire event (with $\OpOf{e} = \acq(\lk)$)
with $f \trord{\tr} e$.
Since $\rho$ is a sync-preserving correct reordering of $\rho$, we must also have
$\match{\tr}(f) \in \events{\rho}$.
Thus, $\rho$ must contain $\SPClosure{\tr}(S)$.

Let us now argue that we can construct a sync-preserving correct reordering whose events are precisely $\SPClosure{\tr}(S)$.
Consider the sequence of events $\rho$ obtained by projecting $\tr$
to $\SPClosure{\tr}(S)$.
We will argue that (a) $\rho$ is a well-formed trace, 
(b) $\rho$ is a correct reordering of $\tr$, and 
(c) $\rho$ preserves the order of critical section on the same lock as in $\tr$.
For (a), we need to argue that critical sections on the same lock do not overlap in $\rho$.
This follows because if they did, then we have two acquire events $e_1 \neq e_2 \in \SPClosure{\tr}(S)$
such that $\OpOf{e_1} = \OpOf{e_2} = \acq(\lk)$ (for some $\lk$),
$e_1 \trord{\tr} e_2$ (and thus $e_1$ also appears before $e_2$ in $\rho$ by construction),
but $\match{\tr}(e_1)$ either does not appear in $\rho$ or appears later than $e_2$.
In the first case, we arrive at a contradiction to the definition of $\SPClosure{\tr}(S)$,
and in the second case, we arrive at a contradiction that $\rho$ is obtained without changing the relative order of the residual events.
Arguing (b) is straightforward: $\SPClosure{\tr}(S)$ is closed under both $\tho{\tr}$ as well as $\rf{\tr}$,
and further $\rho$ preserves the order of events as in $\tr$.
Likewise,  (c) is easy to argue because $\events{\rho} = \SPClosure{\tr}(S)$ are closed
under sync-preservation and the order of events in $\rho$ does not change with respect to $\tr$. 
\end{proof}

Let us now turn our attention to \lemref{spclosure-deadlock}.

\spclosureDeadlock*


\begin{proof}
($\Rightarrow$). Assume $D$ is a sync-preserving deadlock of $\tr$.
Let $\rho$ be the sync-preserving correct reordering of $\tr$ that witnesses $D$.
Thus, each of $e_0, \ldots, e_{k-1}$ are enabled in $\rho$.
This means that $\prev{\tr}(S) \subseteq \events{\rho}$, and also $S \cap \events{\rho} = \emptyset$.
By \lemref{spclosure-necessary-sufficient}, we have
$\SPClosure{\tr}(\prev{\tr}(S)) \subseteq \events{\rho}$, giving us
$\SPClosure{\tr}(\prev{\tr}(S)) \cap S = \emptyset$.

($\Leftarrow$). Assume $\SPClosure{\tr}(\prev{\tr}(S)) \cap S = \emptyset$.
From \lemref{spclosure-necessary-sufficient}, there exists a sync-preserving correct reordering
$\rho$ with $\prev{\tr}(S) \subseteq \SPClosure{\tr}(\prev{\tr}(S)) = \events{\rho}$.
Since $S \cap \events{\rho} = \emptyset$, the entire deadlock pattern $D$ is also enabled in $\rho$.
Thus, $D$ is a sync-preserving deadlock of $\tr$.
\end{proof}

The proof of the following proposition follows directly from the definition of sync-preserving closure.

\spclosureMonotone*

\begin{proof}
Follows from the definition of $\SPClosure{\tr}(S)$.
\end{proof}
\newpage

\section{Incomparability of $\SyncPDOffline$ and \seqc}
\applabel{incomp}
In this section we provide two example traces and discuss the deadlocks found and missed on them by \seqc and \SyncPDOffline.
\exref{incomp-ex1} discusses a trace which contains a deadlock found by \SyncPDOffline but missed by \seqc.
\exref{incomp-ex2} discusses a trace which contains a deadlock found by \seqc but missed by \SyncPDOffline.
These examples showcase the incomparability of the techniques \SyncPDOffline and \seqc.

\begin{example}
\exlabel{incomp-ex1}
\figref{incomp1} displays a trace $\sigma$ in which \SyncPDOffline reports a deadlock whereas \seqc misses the deadlock.
In this trace, there is a single deadlock pattern $\pattern{e_4, e_{14}}$.
This deadlock pattern can be realized by the sync-preserving correct reordering $\rho = e_3, e_{12}, e_{13}, e_8, e_9$.
However, this deadlock is missed by \seqc.
The reason for this can be explained as follows.
First, observe that a correct reordering of $\sigma$ would need to include the event $e_{13}$ and consequently $e_9 = \rf{\sigma}(e_{13})$.
A strategy that \seqc employs when it is producing correct reorderings is that it aims to close all the critical sections that appear in the reordering.
In order to fulfill this constraint, \seqc closes the critical section on lock $\LockColorOne{\lk_1}$ in thread $t_1$.
This requires \seqc to execute the read event $e_{10}$ as well as the event $e_7 = \rf{\sigma}(e_{10})$.
This in turn causes \seqc to execute the events $e_3\ldots e_6$.
Observe that the first event that appears in the deadlock pattern, $e_4$, can no longer be scheduled concurrently with $e_{14}$.
Hence, \seqc is unable to report a deadlock in this trace.
\end{example}

\begin{example}
\exlabel{incomp-ex2}
\figref{incomp2} display a trace $\sigma$ in which \seqc reports one more deadlock than \SyncPDOffline.
In this trace there are two deadlock patterns $\pattern{e_2, e_{6}}$ and $\pattern{e_2, e_{8}}$.
The first deadlock pattern $\pattern{e_2, e_{6}}$ can be realized to a real deadlock by both \seqc and \SyncPDOffline. 
The second deadlock pattern  $\pattern{e_2, e_{8}}$ cannot be realized to a real deadlock by \SyncPDOffline.
Whereas, \seqc is able to report this as a real deadlock.
The reason why \SyncPDOffline misses the second deadlock can be explained as follows.
The computed sync-preserving closure for the second deadlock pattern is $\SPClosure{\tr}(\prev{\tr}(\{ e_{2},e_{8}\})) = \set{e_1,\ldots, e_4,e_5,\ldots,e_7}$.
Since the \SyncPDOffline algorithm cannot reverse the order of critical sections on a given lock, 
it always enforces to close the critical section on lock $\LockColorOne{\lk_1}$ in thread $t_2$ after the one in $t_1$.
This brings the event $e_2$ into the sync-preserving closure of $\pattern{e_2, e_{8}}$
and disables \SyncPDOffline from realizing this deadlock pattern.
\end{example}


\begin{figure}[t]
\centering
\scalebox{1}{
\execution{4}{
\figev{4}{$\acq(\LockColorOne{\lk_1})$}
\figev{4}{$\rel(\LockColorOne{\lk_1})$}
\figev{3}{$\acq(\LockColorTwo{\lk_2})$}
\figev{3}{$\mathbf{\Bacq(\LockColorThree{\lk_3})}$}
\figev{3}{$\rel(\LockColorThree{\lk_3})$}
\figev{3}{$\rel(\LockColorTwo{\lk_2})$}
\figev{3}{$\wt(y)$}
\figev{1}{$\acq(\LockColorOne{\lk_1})$}
\figev{1}{$\wt(x)$}
\figev{1}{$\rd(y)$}
\figev{1}{$\rel(\LockColorOne{\lk_1})$}
\figev{2}{${\acq(\LockColorThree{\lk_3})}$}
\figev{2}{$\rd(x)$}
\figev{2}{$\mathbf{\Bacq(\LockColorTwo{\lk_2})}$}
\figev{2}{$\rel(\LockColorTwo{\lk_2})$}
\figev{2}{$\rel(\LockColorThree{\lk_3})$}
}
}
\caption{
A trace $\tr$ with a sync-preserving deadlock, stalling $t_2$ on $e_{14}$ and $t_3$ on $e_{4}$.
\figlabel{incomp1}
}
\end{figure}

\begin{figure}[t]
\centering
\scalebox{1}{
\execution{2}{
\figev{1}{$\acq(\LockColorOne{\lk_1})$}
\figev{1}{$\mathbf{\Bacq(\LockColorTwo{\lk_2})}$}
\figev{1}{$\rel(\LockColorTwo{\lk_2})$}
\figev{1}{$\rel(\LockColorOne{\lk_1})$}
\figev{2}{${\acq(\LockColorTwo{\lk_2})}$}
\figev{2}{$\mathbf{\Bacq(\LockColorOne{\lk_1})}$}
\figev{2}{$\rel(\LockColorOne{\lk_1})$}
\figev{2}{$\mathbf{\Bacq(\LockColorOne{\lk_1})}$}
\figev{2}{$\rel(\LockColorOne{\lk_1})$}
\figev{2}{$\rel(\LockColorTwo{\lk_2})$}
}
}
\caption{
A  trace $\tr$ with a sync-preserving deadlock, stalling $t_1$ on $e_{2}$ and $t_2$ on $e_{6}$ and a predictable but not sync-preserving deadlock stalling $t_1$ on $e_{2}$ and $t_2$ on $e_{8}$.
\figlabel{incomp2}
}
\end{figure}

\section{Unsoundness Of \dirk}
\applabel{unsound-dirk}
In this section we provide two example programs that showcase the unsoundness of \dirk.
\exref{dirk-unsound1} discusses a program which is a modified version of the Transfer benchmark.
\dirk reports a false deadlock in this example due to that it is aiming to relax the standard definition of correct
reorderings and allow a read event to read from a different write than in the original trace.
Note that this relaxation enables \dirk to predict deadlocks in the benchmarks Transfer and Deadlock, which are missed by $\SyncPDOffline$ and \seqc. 
However, this relaxation is not fully formalized and the implementation of it contains errors.
\exref{dirk-unsound2} discusses a crafted program in which \dirk is reporting a false deadlock due to missing that a common lock protects an otherwise deadlock pattern.

\begin{example}
\exlabel{dirk-unsound2}
	\figref{dirk-unsound-2} displays a program in which \dirk makes a false deadlock report. 
	The program spawns three threads \texttt{T1}, \texttt{T2}, and \texttt{T3} and  
	there exist three locks \texttt{L1}, \texttt{L2}, and \texttt{L3} which are used in these threads.
	In the program, \texttt{T1} acquires the lock \texttt{L1} and then forks \texttt{T2}, and it holds the lock \texttt{L1} until \texttt{T2} is joined.
	Even though the locks \texttt{L2} and \texttt{L3} are acquired in a cyclic manner 
	in the threads  \texttt{T2} and \texttt{T3}, 
	this program does not exhibit a deadlock due to the common lock \texttt{L1}.
	Therefore, this program does not exhibit a deadlock.
	However, \dirk fails to recognize the common lock \texttt{L1} and reports a false deadlock.
	This observed behavior is not an implementation error, but it is in alignment with the underlying theory of \dirk~\cite{Kalhauge2018}.
	In particular, \dirk makes SMT queries to decide the existence of deadlocks, but its encoding unsoundly neglects the condition that deadlocking events must have disjoint lock sets.  
\end{example}

\begin{example}
\exlabel{dirk-unsound1}
	\figref{dirk-unsound-1} displays a program in which \dirk makes a false deadlock report. 
	This program is a modified version of the Transfer benchmark, which was used as a running example to explain \dirk~\cite{Kalhauge2018}.
	The example contains the classes \texttt{Transfer1}, \texttt{Transfer2}, \texttt{Store} and \texttt{UserObject}. 
	The \texttt{Store} class contains two synchronized methods, \texttt{transferTo} and \texttt{add}.
	The \texttt{transferTo} takes another \texttt{Store} object as argument and calls the \texttt{add} method by using this object.
	Notice that if two different \texttt{Store} objects simultaneously call the \texttt{transferTo} 
	method by passing each other as the argument, then this can potentially lead to a deadlock.
	Observe that in the given program there exists such a scenario where \texttt{transferTo} is called 
	with opposite \texttt{Store} objects in the classes \texttt{Transfer1} and \texttt{Transfer2}.
	However, in this case the program can never actually deadlock.
	The reason is that the \texttt{transferTo} method in \texttt{Transfer2} class can only be executed if the \texttt{UserObject.data} can be cast to an integer value.
	The \texttt{UserObject.data} is a volatile variable that is shared by the classes \texttt{Transfer1} and \texttt{Transfer2}. 
	This variable is initialized to a string value and an integer value is assigned to this variable only after the
	\texttt{transferTo} method has been executed in \texttt{Transfer1}. 
	Hence, this program does not exhibit a deadlock.
	However, \dirk reports a deadlock in this program.
	The reason is that \dirk unsoundly concludes that reading the value of \texttt{UserObject.data} at Line $54$ does not affect the control flow of the program.
	Hence, in case the value of \texttt{UserObject.data} object in Line $54$ is read from the write event on Line $39$, 
	\dirk ignores this dependency and concludes that Lines $38$ and $54$ can execute concurrently and reports a false deadlock. 
\end{example}


\begin{figure}[h!]
	\hspace{1cm}
	\begin{subfigure}[t]{.43\linewidth}
\begin{lstlisting}[language = Java,firstline=0,lastline=32]
public class FalseDeadlock1 {
	public static Object L1 = new Object();
	public static Object L2 = new Object();
	public static Object L3 = new Object();
	public static long x, y;
	public static final long ITERS = 10;
	
	public static void main (String [] args) {
		x = 0;
		y = 0;
		T1 t1 = new T1();
		T3 t3 = new T3();
		t1.start();
		t3.start();
		try{
			t1.join();
		}
		catch(Exception ex) { }
		try{
			t3.join();
		}
		catch(Exception ex) { }
	}
	
	static class T1 extends Thread {
		public void run () {
			synchronized (L1) {
				T2 t2 = new T2();
				t2.start();
				try{

					t2.join();
\end{lstlisting}
	\end{subfigure}
	\hfill
	\begin{subfigure}[t]{.45\linewidth}
\begin{lstlisting}[language = Java,firstline=32,lastline=65,firstnumber=33]
					t2.join();
				}
				catch(Exception ex){
				}
			}
		}
	}
	static class T2 extends Thread {
		public void run () {
			synchronized (L2) {
				synchronized (L3) {
					for(int i = 0; i < ITERS; i++){
						x = x + 1;
					}
				}
			}
		}
	}
	static class T3 extends Thread {
		public void run () {
			synchronized (L1) {
				synchronized (L3) {
					synchronized (L2) {
						for(int i = 0; i < ITERS; i++){
							y = y + 1;
						}
					}
				}
			}
		}
	}
}
\end{lstlisting}
	\end{subfigure}
	\caption{FalseDeadlock1.java}
	\figlabel{dirk-unsound-2}
\end{figure}

\begin{figure}[h!]
	\hspace{1cm}
	\begin{subfigure}[t]{.43\linewidth}
\begin{lstlisting}[language = Java,firstline=0,lastline=41]
public class UserObject {
	public volatile Object data;
	
	public UserObject() {
		this.data = "string1";
	}
}

class Store {
	private int var;
	
	public Store(int var) {
		this.var = var;
	}
	
	public synchronized void transferTo(int amount, Store other, int intVal) {
		this.var -= amount;
		other.add(amount);
	}
	
	public synchronized void add(int amount) {
		this.var += amount;
	}
}

public class Transfer1 implements Runnable {
	private Store a;
	private Store b;
	private UserObject uObject;
	
	public Transfer1(Store a, Store b, UserObject uObject) {
		this.a = a;
		this.b = b;
		this.uObject = uObject;
	}
	public void run() {        
		uObject.data = "string2";        
		a.transferTo(1, b, 0);
		uObject.data = 1;
	}
}
\end{lstlisting}
	\end{subfigure}
\hfill
	\begin{subfigure}[t]{.43\linewidth}
\begin{lstlisting}[language = Java,firstline=42,lastline=82,firstnumber=42]
public class Transfer2 implements Runnable {
	private Store a, b;
	private UserObject uObject;
	
	public Transfer2(Store a, Store b, UserObject uObject) {
		this.a = a;
		this.b = b;
		this.uObject = uObject;
	}
	
	public void run() {     
		try {
			a.transferTo(1, b,(int) uObject.data);
		} catch(Exception e) { }
	}
}

public class FalseDeadlock2 {
	public static void main(String[] args) {
		Store a = new Store(1),
		b = new Store(1);
		
		UserObject uObject = new UserObject();
		Transfer1 trans1 = new Transfer1(a, b, uObject);
		Transfer2 trans2 = new Transfer2(b, a, uObject);
		
		Thread[] threads = {
			new Thread(trans1),
			new Thread(trans2)
		};
		
		for (Thread t : threads) t.start();

		for (Thread t : threads)
		try { t.join(); }
		catch (InterruptedException e) { };
	}
}
\end{lstlisting}
	\end{subfigure}
	\caption{FalseDeadlock2.java}
	\figlabel{dirk-unsound-1}
\end{figure}

\end{arxiv}

\end{document}